\documentclass[a4paper,UKenglish,cleveref, autoref, thm-restate,nolineno]{lipics-v2021}

\usepackage{todonotes}
\usepackage{algorithm}
\usepackage[noend]{algpseudocode}
\usepackage{xspace}

\usepackage{thm-restate}
\usepackage{enumitem}

\usepackage{tikz}
\usetikzlibrary{cd}

\def\RR{{\mathbb R}}
\def\NN{{\mathbb N}}
\def\RRR{{\mathcal R}}
\def\GG{{\mathbb G}}
\def\e{{\mathrm e}}
\def\BBB{{\mathcal B}}
\DeclareMathOperator{\dtw}{d_{DTW}}

\DeclareMathOperator{\cost}{cost}

\DeclareMathOperator*{\argmin}{arg\,min}
\DeclareMathOperator{\pexpected}{E}

\def\mean{mean\xspace}%definition 6
\def\Mean{Mean\xspace}

\let\epsilon\relax
\newcommand{\epsilon}{\varepsilon}

\allowdisplaybreaks

\title{Approximating Length-Restricted Means under Dynamic Time Warping
} %TODO Please add

\author{Maike Buchin}{Faculty of Informatics, Ruhr University Bochum,   Germany}{maike.buchin@rub.de}{https://orcid.org/0000-0002-3446-4343}{}
\author{Anne Driemel}{Hausdorff Center for Mathematics, Germany \and  University of Bonn, Germany}{driemel@cs.uni-bonn.de}{}{}
\author{Koen van Greevenbroek}{Department of Computer Science, UiT The Arctic University of Norway, Tromsø, Norway}{koen.v.greevenbroek@uit.no}{https://orcid.org/0000-0002-6105-2846}{}
\author{Ioannis Psarros}{Athena Research Center,  Greece.  
}{ipsarros@di.uoa.gr}{https://orcid.org/0000-0002-5079-5003}{The author was partially supported by the European Union's Horizon 2020 Research and Innovation programme, under the grant agreement No. 957345: “MORE”. Part of this work was done while the author was a postdoctoral researcher at the Hausdorff Center for Mathematics and the  University of Bonn, Germany.}
\author{Dennis Rohde}{Faculty of Informatics, Ruhr University Bochum,  Germany}{dennis.rohde-t1b@rub.de}{https://orcid.org/0000-0001-8984-1962}{}

\authorrunning{M. Buchin and A. Driemel and K. van Greevenbroek and I. Psarros and D. Rohde} 

\Copyright{M. Buchin and A. Driemel and K. van Greevenbroek and I. Psarros and D. Rohde}

 %TODO mandatory; please add comma-separated list of keywords

\nolinenumbers %uncomment to disable line numbering

\hideLIPIcs  %uncomment to remove references to LIPIcs series (logo, DOI, ...), e.g. when preparing a pre-final version to be uploaded to arXiv or another public repository

\EventEditors{John Q. Open and Joan R. Access}
\EventNoEds{2}
\EventLongTitle{42nd Conference on Very Important Topics (CVIT 2016)}
\EventShortTitle{CVIT 2016}
\EventAcronym{CVIT}
\EventYear{2016}
\EventDate{December 24--27, 2016}
\EventLocation{Little Whinging, United Kingdom}
\EventLogo{}
\SeriesVolume{42}
\ArticleNo{23}
\begin{document}
\setlength{\parindent}{0pt}

\maketitle

\begin{abstract}
    We study variants of the \mean problem under the $p$-Dynamic Time Warping ($p$-DTW) distance, a popular and robust distance measure for sequential data. In our setting we are given a set of finite point sequences over an arbitrary metric space and we want to compute a \mean point sequence of given length that minimizes the sum of $p$-DTW distances, each raised to the $q$\textsuperscript{th} power, between the input sequences and the \mean sequence. In general, the problem is $\mathrm{NP}$-hard and known not to be fixed-parameter tractable in the number of sequences. 
    On the positive side, we show that restricting the length of the \mean sequence significantly reduces the hardness of the problem. We give an exact algorithm running in polynomial time for constant-length \mean{s}. We explore various approximation algorithms that provide a trade-off between the approximation factor and the running time. Our approximation algorithms have a running time with only linear dependency on the number of input sequences. In addition, we use our \mean algorithms to obtain clustering algorithms with theoretical guarantees.
\end{abstract}
\clearpage

\thispagestyle{plain}
%\clearpage
\setcounter{page}{1}

\section{Introduction}

The $p$-Dynamic Time Warping distance -- in short $p$-DTW -- is a popular distance measure for temporal data sequences, which has been studied and applied extensively over the past decades. While it was first applied to speech recognition \cite{1163055}, it also showed effective for other kinds of sequential data and now there is a broad variety of applications in numerous domains, cf.~\cite{DBLP:conf/kdd/BerndtC94,731723,DBLP:conf/iccv/MunichP99,DBLP:journals/bioinformatics/AachC01,DBLP:conf/sigmod/ZhuS03,muda2010voice,DBLP:conf/chi/LucaHBLH12}. Its particular strength is the ability to handle differences in the length and in the temporal properties (e.g.~phase or sampling rates) of the data. Furthermore, it is not sensitive to outliers in the sequences, e.g.~from measurement errors or noise. It is based on monotonic alignments of the sequences, i.e., every element of the first sequence is paired to an element of the second sequence in a monotonic fashion (along the temporal axis). To compensate differences in length, samplings rates or in phase, several elements of the first sequence can be paired to a single element of the second sequence and vice versa. Such an alignment is called a warping and the $p$-DTW is the $p$\textsuperscript{th} root of the sum of distances, each raised to the $p$\textsuperscript{th} power, between all pairs of elements determined by an optimal warping, i.e., a warping that minimizes said quantity. The distances between elements are determined by an underlying space, which itself is determined by the application at hand. The $p$-DTW can be computed by a dynamic program with running time quadratic in the lengths of the given sequences, and it cannot be computed in strongly subquadratic time unless the Exponential Time Hypothesis is false \cite{DBLP:conf/focs/BringmannK15,DBLP:conf/focs/AbboudBW15}. Apart from its apparent benefits, $p$-DTW has the drawback that it is not a metric since it does not fulfill the identity of indiscernibles nor the triangle inequality. This rules out a wealth of techniques that were developed for proper metrics. 

In this work, we consider the problem of computing a \mean under $p$-DTW. Here, we are given a finite set of $n$ point sequences over an arbitrary metric space, each of complexity, i.e., the number of elements of the sequence, bounded by a number $m$ and we want to compute a point sequence (the \mean) that minimizes the sum of $p$-DTW distances, each raised to the $q$\textsuperscript{th} power, between the given sequences and the \mean.\footnote{For $q = 1$ this is an adaption of the Euclidean median and for $q = 2$ an adaption of the Euclidean mean.}  We call this problem the \emph{unrestricted} $(p,q)$-\mean problem. It is known to be $\mathrm{NP}$-hard \cite{BFN20,BDS20} and all known algorithms that solve it either suffer from exponential running time, only work for binary alphabets, or are of heuristic nature \cite{PKG11,BFFJNS18,SFN20}. 

We show that when we restrict the complexity of the \mean to be bounded by some constant~$\ell$ -- we call this the \emph{restricted} $(p,q)$-\mean problem -- the problem becomes tractable, i.e., there exist polynomial time approximation algorithms. 
This restriction also comes with a practical motivation, i.e., to suppress overfitting, see also the discussion in~\cite{DBLP:conf/gis/BrankovicBKNPW20}. %% CDTW Paper reference

%, a problem specific for sequential data. 

%Variants of this problem have also been studied under the Hausdorff and Fr\'echet distances, cf.~\cite{DBLP:conf/soda/DriemelKS16,abhin2020kmedian,doi:10.1137/1.9781611976465.160}, which are related distance measures. 

\subsection{Related Work}

%\says{ioannis}{other results that may be worth mentioning:} \cite{JS20}
Among many practical approaches for the problem of computing a \mean, one very influential heuristic is the DTW Barycentric Average (DBA) method, as formalized by  Petitjean, Ketterlin and Gan{\c{c}}arski~\cite{PKG11}. The core idea behind DBA is a Lloyd's style ($k$-means) iterative strategy, which has been rediscovered many times for this problem in the past (see e.g.~\cite{RW79,ACS03,HNF08}). DBA iteratively improves the solution as follows: given a candidate average sequence $c=(c_1,\ldots,c_{\ell})$, it first computes the warpings between $c$ and all input sequences, and then given each set of input vertices $S_i$  matched with the same vertex $c_i$, it substitutes $c_i$ with the mean  of $S_i$. DBA has inspired many recent solutions that are successful in practice~\cite{SJ18,MAKD18,LZZ19,DKP20,O21}. However, it does not give any guarantees. Just like the $k$-means algorithm, it may even converge to a local optimum that is arbitrarily far from the global optimum in terms of the $(p,q)$-mean target function.

There are few results in the literature with formal guarantees on the running time or the quality of the solution. Brill et al.~\cite{BFFJNS18, DBLP:journals/datamine/BrillFFJNS19} presented an algorithm for solving the unrestricted $(2,2)$-\mean problem defined over $\mathbb{Q}$ with the Euclidean distance, with an asymptotic bound on the time complexity. Their algorithm is based on dynamic programming, and computes the (unrestricted) $(2,2)$-\mean, in time $O(m^{2n+1}2^n n)$. The algorithm can be slightly modified, to compute a restricted $(2,2)$-\mean. %The running time then becomes $O(m^{2n+1}2^n n\ell^2)$.
% Koen: some nitpicky comments that probably don't need to be included here. The dynamic program might not work for arbitrary metric spaces. It was shown by Brill et al. over $\RR$, and in fact works whenever you can compute 1-medians in the underlying metric space (with the distance raised to the $p$. Also the factor of $2^n$ can be removed from the time complexity with a slight improvement.
Brill et al.~\cite{BFFJNS18, DBLP:journals/datamine/BrillFFJNS19} also show that the unrestricted $(2,2)$-\mean problem defined over $\{0,1\}$ with the Euclidean distance can be solved in $O(nm^3)$ time. This was later improved by Schaar, Froese and Niedermeier~\cite{SFN20} to $O(nm^{1.87})$ time. 

All previous hardness results concern the exact computation of the $(p,q)$-\mean. Bulteau, Froese and Niedermeier~\cite{BFN20} proved that the $(2,2)$-\mean problem defined over $\mathbb{Q}$ with the Euclidean distance is $\mathrm{NP}$-hard and $\mathrm{W}[1]$-hard with the number of input sequences $n$ as the parameter. Moreover, they show that the problem cannot be solved in time $O(f(n))\cdot m^{o(n)}$ for any computable function $f$ unless the Exponential Time Hypothesis (ETH) fails. Buchin, Driemel and Struijs~\cite{BDS20} presented an alternative proof of the above statements, which more generally applies to the unrestricted $(p,q)$-\mean problem for any $p,q\in \NN$. %We note that their reduction uses inputs over a ternary alphabet while the reduction by Bulteau, Froese and Niedermeier only uses inputs over a binary alphabet, though.

\subsection{Overview of Results}

In this section we give an overview of our results.\footnote{An earlier version of this manuscript claimed hardness of approximation for the problem of computing the mean under dynamic time warping. However, the proof turned out to be flawed. We leave it as an open problem to show hardness of approximation for this problem. }
In Section~\ref{sec:exact_computation}, we present an exact algorithm with polynomial running time for the problem of computing a restricted $2$-\mean in Euclidean space. Our approach is based on a decomposition of the solution space by an arrangement of polynomial surfaces such that each cell corresponds to a set of \mean{}s with uniquely defined optimal warpings to all input sequences. The algorithm has running time  $(nm)^{2^{O(d\ell)}}$, where the input are $n$ sequences of $m$ points in $\RR^d$. Note that the running time is doubly-exponential in $\ell$ and $d$. In the remainder of the paper, our goal is to improve upon this with the help of approximation and randomization techniques. We will show that linear dependency $n$ and singly-exponential dependency on $\ell$ and $d$ are possible.

%While our algorithm has  polynomial running time assuming that  $d$ and $\ell$ are constants, the exact dependence on $d, \ell$ is doubly exponential. Our next results indicate that the running time can be improved to single exponential on $d,\ell$, at the cost of computing approximate solutions instead of exact. 

In \cref{section:constantfactor}, we present a randomized constant-factor approximation algorithm for the restricted $p$-\mean problem that works for sequences from any fixed metric space. As such, this result is applicable to the classical median problem under the $1$-DTW distance and the classical mean problem under the $2$-DTW distance. The main idea is to uniformly sample from the union of points of all given sequences, then enumerate all sequences of complexity $\ell$ from the sampled vertices and return the sequence with the lowest cost. We also show how to derandomize the algorithm.

%If the underlying metric space is a Euclidean space, then we can derandomize our algorithm, at --  perhaps surprisingly -- almost no extra cost in the asymptotic running time. 

In \cref{section:onepluseps} we present a $(1+\epsilon)$-approximation algorithm for the restricted $(p,1)$-\mean problem. This algorithm is based on an exhaustive search over a carefully constructed set of candidate \mean{s}. The crucial ingredients for this are presented in \cref{section:simpltrineq}: an efficient approximation algorithm for simplification under $p$-DTW and a weak triangle inequality for $p$-DTW.
The result holds for the important case of sequences that stem from a Euclidean space. 
A nice property of this result is that it provides a complete trade-off between the approximation factor and the running time. %For example, if we allow a large approximation factor (using $\epsilon=m$), then we achieve a running time of $O(m^4+nm\log m)$ for any constant probability of success. 

Finally, in Section~\ref{sec:clustering}, we briefly discuss an application of the newly developed techniques to the problem of clustering for $p$-DTW distances. In particular, we can use the random sampling techniques developed in  \cref{section:constantfactor,section:simpltrineq}, in combination with a known algorithmic scheme that reduces the computation of $k$-medians to a problem of computing $1$-median candidates for an arbitrary subset of the input set. 
%The results can be found in the appendix, in Section~\ref{appendix:clustering}. 

\subsection{Preliminaries}

In the following $d \in \mathbb{N}$ is an arbitrary constant. For $n \in \mathbb{N}$ we define $[n] = \{1, \dots, n\}$. Let $\mathcal{X} = (X,\rho)$ be a metric space. For $x \in X$ and $r \in \mathbb{R}_{\geq 0}$ we denote by $B(x,r) = \{ y \in X \mid \rho(x,y) \leq r \}$ the ball of radius $r$ centered at $x$. We define sequences of points over $\mathcal{X}$.

\begin{definition}
    A point sequence over $\mathcal{X}$ is a tuple $(\sigma_1, \dots, \sigma_m) \in X^m$, where $m \in \mathbb{N}_{> 1}$ is called its complexity, denoted by $\lvert \sigma \rvert$, and $\sigma_1, \dots, \sigma_m$ are called its vertices.  
\end{definition}
By $X^\ast = \bigcup_{i=1}^\infty X^i$ we define the set of all point sequences over $\mathcal{X}$ and by $X^{\leq m} = \bigcup_{i=1}^{m} X^i$ we define the subset of point sequences of complexity at most $m$. The \emph{concatenation} of a point sequence $\pi=(\pi_1,\ldots,\pi_m)$ with a sequence $\tau=(\tau_1,\ldots,\tau_m)$ is denoted by $\pi \oplus \tau$ and is defined as the point sequence $(\pi_1,\ldots,\pi_m , \tau_1,\ldots,\tau_m)$. 
%In order to simplify operations, we define the concatenation between a point sequence $\pi$ and an “empty” point sequence $(\cdot)$ such that  $\pi \oplus (\cdot)=\pi $. 
We define the $p$-Dynamic Time Warping distance.
\begin{definition}
    \label{def:pq_dtw}
    For $m_1, m_2 \in \mathbb{N}_{> 1}$, let $\mathcal{W}_{m_1,m_2}$ denote the set of all $(m_1, m_2)$-\emph{warpings}, that is, the set of all sequences $(i_1, j_1), \dots, (i_n, j_n)$ with
    \begin{itemize}
        \item $i_1 = j_1 = 1$, $i_n = m_1$, $j_n = m_2$ and
        \item $(i_{k} - i_{k-1}, j_{k} - j_{k-1}) \in \{(0,1), (1,0), (1,1)\}$ for each $k \in \{2, \dots, n\}$.
    \end{itemize}
    For $p \in [1, \infty)$ and two point sequences $\sigma = (\sigma_1, \dots, \sigma_{m_1}) \in X^{m_1}, \tau = (\tau_1, \dots, \tau_{m_2}) \in X^{m_2}$ the $p$-Dynamic Time Warping distance is defined as $$ \dtw_p(\sigma, \tau) = \min_{W \in \mathcal{W}_{m_1,m_2}} \left( \sum_{(i,j) \in W} \rho(\sigma_i, \tau_j)^p \right)^{1/p} . $$
\end{definition}

Here we assume that $\rho(\cdot,\cdot)$ can be evaluated in constant time.\footnote{This restriction is only for the sake of simplicity of presentation. Our results can be easily extended to metric spaces that do not have a constant-time distance oracle.} Let $\sigma = (\sigma_1, \dots, \sigma_{m_1}), \tau = (\tau_1, \dots, \tau_{m_2})$ be point sequences over $\mathcal{X}$ of complexity $m_1$ and $m_2$. We call a warping $W \in \argmin\limits_{W \in \mathcal{W}_{m_1, m_2}} \left( \sum_{(i,j) \in W} \rho(\sigma_i, \tau_j)^p \right)^{1/p}$ an optimal $p$-warping between $\sigma$ and $\tau$. 

\begin{definition}
    The restricted $(p, q)$-\mean problem is defined as follows, where $\ell \in \mathbb{N}_{>1}$ and $p,q \in [1, \infty)$ are fixed (constant) parameters of the problem: given a set $T = \{ \tau_1, \dots, \tau_n \} \subseteq X^{\leq m}$ of point sequences, compute a point sequence $c \in X^{\leq \ell}$, such that $ \cost_p^q(T,c) = \sum_{i=1}^n \dtw_p(c, \tau_i)^q$ is minimal.
\end{definition}

If $p$ is clear from the context, we drop it from our notation. By the \textit{unrestricted} $(p,q)$-\mean problem we define the problem that is similar to the restricted $(p,q)$-\mean problem with the only difference that we compute a \mean sequence $c \in X^\ast$. We mainly study the case that $p=q$. Therefore, we shorthand call the restricted, respectively unrestricted, $(p,p)$-\mean problem the restricted, respectively unrestricted, $p$-\mean problem. We emphasize that these problems are prevalent in the literature.

\section{Tractability of the Restricted \Mean Problem}

We study exact and approximation algorithms for the restricted $p$-\mean and $(p,1)$-\mean problems.

\subsection{Exact Computation of a Restricted $2$-\Mean in Euclidean Space}
\label{sec:exact_computation}
In the following, we make use of a simplified structure of the solution space, which holds in case $p=q$. This is captured in the notion of sections, which we define as follows.

\begin{definition}[sections]
    \label{def:section}
    Let $T = \{ \tau_1, \dots, \tau_n \} \subseteq X^{\leq m}$ be a set of point sequences and $c = (c_1, \dots, c_{\ell}) \in X^{\ell}$ be a point sequence. For $i \in [n]$ and $p \in [1, \infty)$, let $W_i$ be an optimal $p$-warping between $c$ and $\tau_i$. For $j \in [\ell]$ we define the $j$\textsuperscript{th} section of $c$ with respect to $T$ (and $W_1, \dots, W_n$) as follows: $ S_j(c, T, W_1, \dots, W_n) = \{ \tau_{i,k} \mid i \in [n], (j,k) \in W_i \}$,
    where $\tau_{i,k}$ is the $k$\textsuperscript{th} vertex of $\tau_i$.
\end{definition}

If $T$ is clear from the context, we omit it from the notation. Also, we will always omit $W_1, \dots, W_n$ from the notation, because the specific choice of optimal $p$-warpings is not of interest. We will then write $S_j^p(T,c)$ to clarify that the sections are defined with respect to optimal $p$-warpings. An immediate consequence of this definition is the following identity:
\[ \cost^p_p(T, c) = \sum_{j=1}^\ell \sum_{v \in S^p_j(c,T)} \rho(c_j, v)^p, \] 
where $\ell$ denotes the complexity of $c$.

A central observation is that the vertices of an optimal restricted $p$-\mean $c = (c_1, \dots, c_{\ell'})$ must minimize the sum of distances, each raised to the $p$\textsuperscript{th} power, to the vertices in their section, i.e., for all $j \in [\ell']$: $c_j \in \argmin\limits_{w \in X} \sum\limits_{v \in S^p_j(c,T)} \rho(w, v)^p$. Using this, we obtain the following result.

\begin{restatable}{theorem}{exactmeanalgorithm}
\label{theorem:exactcomputation}
    There exists an algorithm that, given a set $T \subset \left(\mathbb{Q}^d\right)^{\leq m}$ of $n$ point sequences, computes an optimal restricted $2$-\mean (defined over the Euclidean distance) in time $(nm)^{2^{O(d\ell)}}$.
\end{restatable}

The proof is deferred to \cref{appendix:exact}. For the sake of simplicity, suppose that we want to compute the best \mean of complexity exactly $\ell' \in[\ell]$. To compute the optimal restricted $2$-\mean, it suffices to find the best \mean for every  $\ell'\in [\ell]$.  For a fixed $\ell' \in [\ell]$, the main idea is to compute for any two warpings between a point sequence of complexity $\ell'$ and an input point sequence, a polynomial function whose sign indicates which of the warpings yields a smaller distance between the sequences. These functions are then used to define an arrangement that partitions the space $\left(\mathbb{R}^d\right)^{\ell'}$. The trick is that while there is an infinite number of point sequences in $\left(\RR^d\right)^{\ell'}$, to each input point sequence there are only $O(m^{2{\ell'}})$ warpings and in each face of the arrangement the point sequences have the same optimal warpings to the input point sequences. Therefore, for an arbitrary point sequence from each face of the arrangement, we can compute the optimal warpings to the input sequences and then use the resulting sections to compute an optimal point sequence for these warpings, obtaining the optimal \mean of complexity exactly $\ell'$ when we eventually hit the face containing it. We use the cylindrical algebraic decomposition algorithm to compute the arrangement and obtain an element of each face.

\subsection{Constant-Factor Approximation of the Restricted $p$-\Mean}
\label{section:constantfactor}

%Van Greevenbroek also describes a simple approximation algorithm for a variant of the problem in~\cite{koen_2020}.
We start by describing a simple approximation algorithm that reveals the basic idea underlying the following algorithms. The algorithm relies on the following observation. If $p$-DTW is defined over a metric, then the triangle inequality holds for the point-to-point-distances in the sum that defines the $p$-DTW distance (albeit not for $p$-DTW distance itself). Assume for simplicity that $p=1$. In this case, there always exists a $2$-approximate \mean that is formed by points from the input sequences. Enumerating all possible such sequences, then, if the input consists of $n$ point sequences of length $m$,  leads to an algorithm with running time in $O((nm)^{\ell+1})$, where $\ell$ denotes the largest allowed complexity of the \mean. This approach also extends to other variants of the \mean problem for different choices of $p$ and $q$ (with varying approximation factors). % The approximation factor can be improved if the underlying metric space is a doubling space.
One obvious disadvantage of this simple  algorithm is the high running time. In the following, we use similar observations as above and show that the dependency on the number of input sequences $n$ can be improved to linear while still achieving approximation factors close to $2$.

\subsubsection{Randomized Algorithm}

We present a randomized constant-factor approximation algorithm for the restricted $p$-\mean problem. The approximation factor of the algorithm depends on $p$, and the best it can achieve is $2+\epsilon$ for $p=1$ and $4+\epsilon$ for $p=2$, which resemble the famous Euclidean median and mean problems. The idea of the algorithm is to obtain for each $j \in [\ell']$ from the corresponding section $S^p_j(c, T)$ of an optimal restricted $p$-\mean $c = (c_1, \dots, c_{\ell'})$ one of the closest input vertices to $c_j$. The obtained vertices in the corresponding order form an approximate restricted $p$-\mean. We formalize the idea in the following lemma. The proof is deferred to \cref{appendix:constant}.

\begin{restatable}{lemma}{lemmadiscreteapprox}
    \label{lem:discrete_approx}
    Let $T = \{ \tau_1 = (\tau_{1,1}, \dots, \tau_{1,\lvert \tau_1 \rvert}), \dots, \tau_n = (\tau_{n,1}, \dots, \tau_{n,\lvert \tau_n \rvert}) \} \subseteq X^{\leq m}$ be a set of point sequences and let $P = \bigcup_{i=1}^n \bigcup_{j=1}^{\lvert \tau_i \rvert} \{ \tau_{i,j} \}$. For any $p \in [1, \infty)$, $\ell \in \mathbb{N}_{>1}$ and $\epsilon \in \left(0, \infty \right)$ there exists an $\ell'\leq \ell$ and balls $B_1, \dots, B_{\ell'} \subseteq P$, of cardinality at least $\frac{\epsilon n}{2^{p-1} + \epsilon}$ each, such that any point sequence $c^\prime = (c^\prime_1, \dots, c^\prime_{\ell'})$, with $c^\prime_i \in B_i$ for each $i \in [\ell']$, is a $(2^{p}+\epsilon)$-approximate restricted $p$-\mean for~$T$.
\end{restatable}

Now we present the first algorithm. The idea is to uniformly sample from the set of all vertices of all point sequences, to obtain at least one vertex from each ball guaranteed by the previous lemma, with high probability. After the sampling, the algorithm enumerates all point sequences of at most $\ell$ elements from the sample and returns a point sequence with lowest cost.

\begin{algorithm}[H]
\caption{Restricted $p$-\Mean Constant-Factor Approximation \label{alg:1_median_2_3}}
    \begin{algorithmic}[1]
        \Procedure{\mean-C}{$T = \{\tau_1 = (\tau_{1,1}, \dots, \tau_{1,\lvert \tau_1 \rvert}), \dots, \tau_n = (\tau_{n,1}, \dots, \tau_{n,\lvert \tau_n \rvert}) \}, \delta, \epsilon, p$}
            \State $P \gets \bigcup_{i=1}^n \bigcup_{j=1}^{\lvert \tau_i \rvert} \{ \tau_{i,j} \}$
            \State $S \gets$ sample $\left\lceil \frac{m(\ln(\ell) + \ln(1/\delta))}{\epsilon/(2^{p-1} + \epsilon)} \right\rceil$ points from $P$ uniformly and independently at random
            
            \hspace{\algorithmicindent} with replacement
            \State $C \gets S^{\leq\ell}$
            \State \Return an arbitrary element from $\argmin\limits_{c \in C} \cost_p^p(T,c)$
        \EndProcedure
    \end{algorithmic}
\end{algorithm}
The correctness of \cref{alg:1_median_2_3} follows by an application of \cref{lem:discrete_approx}. The proof of the following theorem can be found in \cref{appendix:constant}. 
\begin{restatable}{theorem}{thmmedianconstantfactor}
\label{thm:1_median_2_3}
    Given a set $T = \{\tau_1, \dots, \tau_n\} \subseteq X^{\leq m}$ of point sequences, three parameters $\delta \in (0,1)$, $\epsilon \in (0, \infty)$ and $p \in [1, \infty)$, \cref{alg:1_median_2_3} returns with probability at least $1-\delta$ a $(2^{p}+\epsilon)$-approximate restricted $p$-\mean for $T$, in time $O\left(nm^{\ell+1} \ln(1/\delta)^\ell \left(1 + \frac{2^{p-1}}{\epsilon}\right)^\ell\right)$.
\end{restatable}

\subsubsection{Derandomization}

In this section, we consider finite metric spaces $(X,\rho)$ for which the set of all metric balls $\mathcal{B} = \{ B(x,r) \mid x \in X, r \in \RR_{\geq 0} \}$ forms a range space $(X,\mathcal{B})$ with bounded VC dimension~$\mathcal{D}$. We present a deterministic algorithm for the restricted $p$-\mean problem which is applicable under the additional assumption that there is a subsystem oracle for $(X, \mathcal{B})$. We show that this is the case for the Euclidean metric. %Interestingly, Huang et al.~\cite{HJLW18} show that if one allows $(1\pm \epsilon)$ distortion on the original distances of a metric space with bounded doubling dimension, then the VC dimension of the range space induced by the metric balls is also bounded as a function of the doubling dimension and $\epsilon$. 
Note that our algorithm depends on the existence of a subsystem oracle, which is not always obvious for a given metric.
\\

We formally define range spaces and the associated concepts. A range space is defined as a pair of sets $(X,\RRR)$, where $X$ is the \textit{ground set} and $\RRR\subseteq 2^{X}$ is the \textit{range set}. For $Y\subseteq X$, we denote $ \RRR_{|Y}= \{ R \cap Y \mid R \in \RRR \} $ and if $\RRR_{|Y} $ contains all subsets of $Y$, then $Y$ is \textit{shattered} by~$\RRR$. A measure of the combinatorial complexity of such a range space is the VC dimension. 
\begin{definition}[VC dimension]
The Vapnik-Chervonenkis dimension~\cite{Sau72,She72,VC71} of $(X,\RRR)$ is the maximum cardinality of a shattered subset of $X$.
\end{definition}
Range spaces need not to be finite and can be discretized by means of $\epsilon$-nets.
\begin{definition}[$\epsilon$-net]
\label{def:epsilon_net}
A set $N \subset X$ is an $\epsilon$-net for $(X,\RRR)$ if for any range $R\in \RRR$, $R \cap N\neq \emptyset$ if $|R\cap X|\geq \epsilon |X| $. 
\end{definition}
To compute $\epsilon$-nets deterministically, we need a subsystem oracle, which we now define. 
\begin{definition}[subsystem oracle]
Let $(X, \RRR)$ be a finite range space. A subsystem oracle is an algorithm which for any $Y \subseteq X$, lists all sets in $\RRR_{|Y}$ in
time $O(|Y|^{\mathcal{D}+1})$, where $\mathcal{D}$ is the VC dimension of $(X, \RRR)$.
\end{definition}
We use the following theorem to obtain $\epsilon$-nets when provided with a subsystem oracle.
\begin{theorem}[{\cite[Theorem 2.1]{BCM99}}]\label{thm:determnet}
Let $(X, \RRR)$ be a range space with finite ground set and VC dimension $\mathcal{D}$, and $ \epsilon > 0$ be a given parameter. Assume that there is a subsystem oracle for $(X, \RRR)$. Then an $\epsilon$-net of size $O\left(\frac{\mathcal{D}}{\epsilon} \log \frac{\mathcal{D}}{\epsilon}\right)$ 
can be computed deterministically
in time $O(\mathcal{D}^{3\mathcal{D}})\cdot \left(\frac{1}{\epsilon} \log  \frac{1}{\epsilon} \right)^\mathcal{D} \cdot |X|$.
\end{theorem}

The following algorithm is a modification of \cref{alg:1_median_2_3} where the sampling step is substituted for a computation of an $(\epsilon/m)$-net of the set of all vertices of all given point sequences. Since the balls guaranteed by \cref{lem:discrete_approx} are of appropriate size, the $(\epsilon/m)$-net stabs all of them and by enumeration of all point sequences of at most $\ell$ points from the $(\epsilon/m)$-net, we again find a good approximate restricted $p$-\mean.

\begin{algorithm}[H]
\caption{Restricted $p$-\Mean Constant-Factor Approximation \label{alg:median_det}}
    \begin{algorithmic}[1]
        \Procedure{\mean-C-D}{$T = \{ \tau_1 = (\tau_{1,1}, \dots, \tau_{1,\lvert \tau_1 \rvert}), \dots, \tau_n = (\tau_{n,1}, \dots, \tau_{n,\lvert \tau_n \rvert}) \}, \epsilon,  p$}
            \State $\epsilon^\prime \gets \frac{\epsilon}{2^{p-1} + \epsilon}$, $P \gets \bigcup_{i=1}^n \bigcup_{j=1}^{\lvert \tau_i \rvert} \{ \tau_{i,j} \}$
            \State $S \gets$ compute an $(\epsilon^\prime/m)$-net of $(P,\mathcal{B})$
            \State $C \gets S^{\leq\ell}$
            \State \Return an arbitrary element from $\argmin\limits_{c \in C} \cost_p^p(T,c)$
        \EndProcedure
    \end{algorithmic}
\end{algorithm}
The correctness of \cref{alg:median_det} follows from \cref{def:epsilon_net}. The proofs of the following statements can be found in \cref{appendix:constant}.
\begin{restatable}{theorem}{thmmedianconstantdetermcor}
    \label{thm:det_mean_correctness}
    Given a set $T \subseteq X^{\leq m}$ of $n$ point sequences, parameters $\epsilon \in (0, \infty)$, and $p \in [1, \infty)$, \cref{alg:median_det} returns a $(2^{p}+\epsilon)$-approximate restricted $p$-\mean for $T$.
\end{restatable}

We now turn to the Euclidean setting. First, we prove that there exists a subsystem oracle for $(X, \mathcal{B})$, when $X \subset \mathbb{R}^d$ is a finite subset of the $d$-dimensional Euclidean space.

\begin{restatable}{lemma}{lemmasubsystemoracle}
    \label{lem:subspaceoracle}
    There is a subsystem oracle for the range space $(X,\mathcal{B})$, where $X$ is a finite subset of $\mathbb{R}^d$.
\end{restatable}

Then we can analyze the running time of \cref{alg:median_det} in the Euclidean setting.

\begin{restatable}{theorem}{thmmedianconstantdeterfinal}
\label{thm:constantfactorderandomized}
    Given a set $T\subset  \left(\RR^d \right)^{\leq m}$ of $n$ point sequences,  parameters $\epsilon \in (0, \infty)$, and $p \in [1, \infty)$, \cref{alg:median_det} can be implemented to run in $O\left( n m \left(\left(\frac{m}{\epsilon^\prime} \log\frac{m}{\epsilon^\prime}\right)^{d+1} + \left(\frac{m}{\epsilon^\prime} \log\frac{m}{\epsilon^\prime}\right)^{\ell} \right) \right)$ deterministic time, where $\epsilon^\prime = \frac{\epsilon}{2^{p-1} + \epsilon}$.
\end{restatable}

\subsection{Simplifications and the Triangle Inequality}
\label{section:simpltrineq}

In this section, we give an efficient approximation algorithm for simplification under $p$-DTW and show that a weak triangle inequality for $p$-DTW holds. These results are then used in \cref{section:onepluseps}, where we provide an algorithm for the restricted $(p,1)$-\mean problem, which achieves an approximation factor of $(1+\epsilon)$, for any $\epsilon \in (0,\infty)$. The result holds for the important case of sequences lying in the Euclidean space.

In this section, we focus on approximating a minimum-error simplification of a point sequence with respect to the DTW distance, and we show that a weak triangle inequality holds. This improves upon previous similar statements~\cite{L09} when one bounds the DTW distance of two short point sequences. We will later apply these two results in the analysis of one of our \mean approximation algorithms. In particular, we will bound the expected cost of a \mean obtained by randomly sampling an input point sequence and then computing its approximate minimum-error simplification.

\subsubsection{Minimum-Error Simplification}
We first define the notion of simplification of a point sequence under the $p$-DTW distance.
\begin{definition}
Let $\pi \in X^\ast$. An $(\alpha,\ell)$-simplification of $\pi$, under the $\dtw_p$ distance, is a point sequence $\tilde{\pi}\in X^{\leq \ell}$ such that 
\[\forall \pi' \in X^{\leq \ell}:\dtw_p(\pi,\tilde{\pi})\leq \alpha \cdot \dtw_p(\pi,{\pi'}).\] 
\end{definition}
We design a dynamic programming solution for the problem of computing a simplification. 
Each subproblem is parameterized by the length of a  prefix of the input point sequence and the maximum length of a simplification of that prefix. 
 A detailed description of our algorithm and the analysis can be found in Section~\ref{appendix:simpltrineq}. 
\begin{restatable}{theorem}{thmsimplification}
\label{theorem:simplification}
     There is an algorithm that given as input $\pi \in X^m$,   computes a $(2,\ell)$-simplification of $\pi$ under the $\dtw_p$, in time $O(m^4\ell)$. 
\end{restatable}

%%%%

\subsubsection{Weak Triangle Inequality}
While DTW is not a metric and it is known that the triangle inequality fails for certain instances, 
 there is a weak version of the triangle inequality that is satisfied. %The following statement, as presented by Lemire~\cite{L09}, is about point sequences of equal complexity. 
 In particular, Lemire~\cite{L09} shows that 
 given  $x, y, z \in X^m$, and $ p \in[1, \infty)$, we
have ${\dtw}_p(x, z) \leq m^{1/p}\cdot \left( {\dtw}_p(x, y) + {\dtw}_p(y, z) \right)$. 
We slightly generalize the above inequality in a way that implies a better bound for the distance between two short point sequences using the distances to a potentially longer point sequence. All missing proofs can be found in \cref{appendix:simpltrineq}. 
\begin{restatable}{lemma}{lemtriangleineq}\label{lem:asymtriangineq2}

 For any $m_1,m_2 \in \NN$, let $x,z\in X^{\leq m_1}$, $y\in X^{m_2}$, and  $p \in[1, \infty) $. Then, 
 \[
 {\dtw}_p(x,z)\leq m_1^{1/p}\cdot \left(  {\dtw}_p(x,y)+{\dtw}_p(y,z) \right).
 \]
\end{restatable}

The following theorem uses the weak triangle inequality and 
provides an upper bound on the expected cost of the restricted $(p,1)$-\mean obtained by first sampling an input point sequence uniformly at random and then computing an $(\alpha,\ell)$-simplification of this point sequence. This theorem will be useful in the next section, where we design an approximation scheme for the \mean problem that relies on a first rough estimation of the cost.

\begin{restatable}{theorem}{lemsampleonesequence}\label{lem:sampleonesequence}
    Let $T=\{\tau_1,\ldots,\tau_n\}\subseteq X^{\leq m}$ be a set of point sequences and let $p \in [1,\infty)$. Let $\pi$ be a point sequence picked uniformly at random from $T$, and let $\tilde{\pi}$ be an $(\alpha,\ell)$-simplification of $\pi$ under $\dtw_p$, 
     where $\ell\leq m$. Then,  
   \begin{align*}
       \pexpected_{\pi} \left[\cost_p^1(T,\tilde{\pi})\right]\leq (2+\alpha)m^{1/p}\ell^{1/p} \cdot \mathrm{OPT}_{\ell},
   \end{align*} 
    %\end{enumerate}
    where $\mathrm{OPT}_{\ell}$ denotes the cost of the optimal restricted $(p,1)$-\mean of $T$. 
\end{restatable}

\begin{proof}
 Let $c$ be an optimal $(p,1)$-\mean of $T$ with cost $\mathrm{OPT}_{\ell}$. 
Then,
\begin{align}
\pexpected_{\pi}[\mathrm{cost}(T,\tilde{\pi})]&=
\pexpected_{\pi}\hspace{-0.2em}\Biggl[\sum_{i=1}^n \dtw_p(\tau_i,\tilde{\pi})\Biggr] \leq \pexpected_{\pi}\hspace{-0.2em}\Biggl[m^{\frac{1}{p}} \sum_{i=1}^n \left( \dtw_p(\tau_i,c) +\dtw_p(c,\tilde{\pi}) \right)\Biggr] \label{eq:trineq1}\\
&=m^{1/p}\cdot (\mathrm{OPT}_{\ell}+n\cdot \pexpected_{\pi}[\dtw_p(c,\tilde\pi)]) \nonumber\\
&\leq m^{1/p}\cdot (\mathrm{OPT}_{\ell}+n\cdot \ell^{1/p} \cdot \pexpected_{\pi}[\dtw_p(c,\pi) + \dtw_p(\pi,\tilde\pi)]) \label{eq:trineq2}\\ 
&\leq  m^{1/p}\cdot (\mathrm{OPT}_{\ell}+(1+\alpha)\cdot n\cdot \ell^{1/p} \cdot \pexpected_{\pi}[\dtw_p(c,\pi)]) \nonumber\\
&=  m^{1/p}\cdot \left(\mathrm{OPT}_{\ell}+(1+\alpha)\cdot n\cdot \ell^{1/p} \cdot \sum_{\pi \in T} \dtw_p(c,\pi)\cdot \frac{1}{n}\right)\nonumber  \\ 
&=  m^{1/p}\cdot (\mathrm{OPT}_{\ell}+(1+\alpha)\cdot \ell^{1/p} \cdot \mathrm{OPT}_{\ell}) \leq (2+\alpha) m^{1/p} \ell^{1/p} \cdot \mathrm{OPT}_{\ell},\nonumber
\end{align}
where in Step (\ref{eq:trineq1}) and Step (\ref{eq:trineq2}), we applied \cref{lem:asymtriangineq2}.
\end{proof}

\subsection{Approximation Scheme for Point Sequences in Euclidean Spaces}
\label{section:onepluseps}

Now we study the restricted $(p,1)$-\mean problem for point sequences in the Euclidean space. This is exactly the problem of computing one median point sequence of complexity at most~$\ell$, under the $p$-DTW distance. 
Formally, input point sequences belong to  $\left(\RR^d\right)^{\leq m}$ and we are interested in computing a median point sequence in $\left(\RR^d\right)^{\ell}$. The distance between any two points $x,y\in \RR^d$ is measured by the Euclidean distance $\|x-y\|$. For any $x\in \RR^d$, $r>0$, $B(x,r)$ denotes the Euclidean ball $\{y\in \RR^d\mid \|x-y\|\leq r \}$. We also use  Euclidean grids:

\begin{definition}[grid]
\label{def:grid}
Given a number $r \in \mathbb{R}_{+}$, for $x = (x_1, \dots, x_d) \in \mathbb{R}^d$ we define by $G(r, x) = (\lfloor x_1/r\rfloor \cdot r, \ldots , \lfloor x_d/r\rfloor \cdot r)$ the $r$-grid-point of $x$. Let $P \subseteq \mathbb{R}^d$ be a subset of $\mathbb{R}^d$. The grid of cell width $r$ that covers $P$ is the set $\GG(P, r) = \{G(r, x) \mid x \in P\}$.
\end{definition}

A grid partitions $\mathbb{R}^d$ into cubic regions. For any $r \in \mathbb{R}_+$, $x \in P$, we have $\lVert x - G(r, x) \rVert \leq r \sqrt{d}$.

\subsubsection{Algorithm}

We build upon ideas developed in Section~\ref{section:simpltrineq} and  
we design a $(1+\epsilon)$-approximation algorithm for the restricted $(p,1)$-\mean problem. The algorithm is randomized and succeeds with probability $1-\delta$, where $\delta$ is a user-defined parameter. 

The high-level idea is the following. Given a set $T$ of $n$ point sequences, we first compute a rough estimate of the optimal cost. To do so, we sub-sample a sufficiently large number of input sequences that we store in a set $S$, and we compute  a  $(2,\ell)$-simplification for each one of  them. We detect a sequence in $S$ whose simplification minimizes the restricted $(p,1)$-\mean cost; this cost is denoted by  $R$. The results of Section~\ref{section:simpltrineq} imply  that with good probability, $R$ is a $O((m\ell)^{1/p})$ approximation of the optimal cost.  
We can now use $R$ to ``guess'' a refined  estimate for the restricted $(p,1)$-\mean cost which is within a constant factor from the optimal, by enumerating multiples of $2$ in the interval  $[\Omega(R(m\ell)^{-1/p}), R]$.   Assuming that we have such an estimate, we can use it to fine-tune a grid, which is then intersected with balls centered at the points of sequences in $S$. We use the resulting grid points to compute a set of candidate solutions. 
The idea here is that with good probability one of the point sequences in $S$ is very close to the optimal solution, so one of the candidate solutions will be a good approximation. 
%Given a set $T$ of $n$ point sequences, we first sample a subset $S\subseteq T$ of $O(\log(1/\delta))$ point sequences uniformly and independently at random with replacement. Among all $(2,\ell)$-simplifications of point sequences in $S$, we detect one with the minimum cost as a median.  We use that cost $R$ to define a set $I_R$ containing $O(\log(m\ell)/p)$ values between $\Omega\left(\frac{R}{ 8n m^{1/p}\ell^{1/p}}\right)$ and $\frac{R}{n}$.  
%For each value $r\in I_R$, and for each point sequence $\tau_i\in S$, we use the grid of side-length $\frac{\epsilon \cdot r}{2m^{1/p}\sqrt{d}}$  to compute a set of candidate medians with vertices lying in balls of radius $4r$ centered at the vertices of $\tau_i$. The output is then the best solution among all candidates, computed for all $r\in I_R$, $\tau_i\in S$. 
The complete  pseudocode can be found in \cref{alg:alpha_apprx_median}.

\begin{algorithm}
\caption{Restricted $(p,1)$-\mean $(1+\epsilon)$-Approximation \label{alg:alpha_apprx_median}}
    \begin{algorithmic}[1]
        \Procedure{Med-Appr}{$T = \{ \tau_1 = (\tau_{1,1}, \dots, \tau_{1,|\tau_1|}), \dots, \tau_n = (\tau_{n,1}, \dots, \tau_{n,|\tau_n|}) \}, \epsilon, p, \delta $}
            \State $S\gets$ sample $\lceil\log(2/\delta)\rceil$ point sequences from $T$ uniformly and independently
            
            \hspace{\algorithmicindent} at random with replacement
            \State $\RRR \gets \emptyset $, $C\gets \emptyset$
            \For{\textbf{each} $\tau_i \in S$}
           % \State $N\gets \emptyset$
            \State $\tau_i' \gets$ $(2,\ell)$-simplification of $\tau_i$ under $\dtw_p$
            \State $\RRR \gets \RRR \cup \{  \cost_p^1(T,\tau_i') \}$
            \EndFor
            \State $R \gets \min \RRR$
            \State $\beta \gets 2 \cdot \left( \frac{68m^{1/p} }{\epsilon} +5 \right)^d$
           % \State $I_R \gets \{\frac{R}{ 8n m^{1/p}\ell^{1/p} }, 2\cdot \frac{R}{8 n m^{1/p}\ell^{1/p} } , 2^2\cdot \frac{R}{ 8 n m^{1/p}\ell^{1/p} }, \ldots, \frac{R}{n}\} $
            \State $I_R \gets \left\{\frac{R\cdot 2^{-i}}{ n }  \mid i=0,\ldots,\lceil 3+ \log(m\ell)/p \rceil \right\} $
            \For{\textbf{each} $r \in I_R$}
            \State $\gamma \gets \frac{\epsilon \cdot r}{(2m)^{1/p}\sqrt{d}}$
            %\State $\beta \gets 2 \cdot \left( \frac{34 r}{\gamma \sqrt{d}}+5 \right)^d$
            \For{\textbf{each} $\tau_i \in S$}
            \State $\mathcal{B}(\tau_{i},4r)\gets \bigcup_{j=1}^{\lvert \tau_i \rvert} B(\tau_{i,j},4r)$
            \State $N \gets \mathbb{G}(\mathcal{B}(\tau_{i},4r),\gamma)$
            \If{$|N|\leq \ell \cdot \beta $}
            \State $C\gets C \cup N^{\leq \ell}$
            \EndIf
            \EndFor
            \EndFor
            \State \Return an arbitrary element of $\argmin\limits_{c \in C} \cost_p^1(T,c)$.
        \EndProcedure
    \end{algorithmic}
\end{algorithm}

\subsubsection{Analysis}
Now we analyze the running time and correctness of \cref{alg:alpha_apprx_median}. Missing proofs can be found in \cref{appendix:onepluseps}. 
We begin with a bound on the probability that $R$ is a rough approximation of the optimal median cost. 
\begin{restatable}{lemma}{lemroughapprox}
\label{lemma:roughapprox}
Let $c$ be an optimal restricted $(p,1)$-\mean of $T$. 
With probability at least $1-\delta/2$, \[R\leq  8m^{1/p}\ell^{1/p} \cdot \cost(T,c).\] 
\end{restatable}

Next, we bound the probability that a point sequence in the sample $S$ is conveniently close to the optimal median.

\begin{restatable}{lemma}{lemneartooptimal}\label{lemma:neartooptimal}
Let $c$ be an optimal restricted $(p,1)$-\mean of $T$. 
With probability at least $1-\delta/2$, there exists a $\tau_i \in S$ such that  $\dtw_p(\tau_i,c)<(2/n)\cdot {\cost(T,c)}$. 
\end{restatable}

The set $I_R$ contains a value $r$ such that $nr$ is within a factor of $2$ from the optimal cost. 
\begin{restatable}{lemma}{lemrightscale}\label{lemma:rightscale}
Let $c$ be an optimal restricted $(p,1)$-\mean of $T$. If 
$ R\leq 8 m^{1/p}\ell^{1/p} \cost(T,c)$, 
then there exists $r\in I_R$ such that 
$\cost(T,c) \in [nr, 2nr]$.
\end{restatable}

The following is an upper bound on the number of grid cells needed to cover a Euclidean ball. Similar bounds appear often in the literature, but they are typically asymptotic and not sufficient for our needs. Therefore, we prove  an exact (non-asymptotic) upper bound. 
\begin{restatable}{lemma}{lemvolumebound}
\label{lemma:volumetricbound}
Let $x\in \mathbb{R}^d$, $r>0$, $\gamma>0$.  
\[
\left|\mathbb{G}\left( B(x,8 r),\gamma\right)\right| \leq 2 \cdot \left( \frac{34 r}{\gamma \sqrt{d}}+5 \right)^d.
\]
 
\end{restatable}

We now focus on the iteration of the algorithm with  $r\in I_R$, $\tau_i\in S$, such that $r$ satisfies the property guaranteed by Lemma~\ref{lemma:rightscale} and $\tau_i$ satisfies the property guaranteed by  Lemma~\ref{lemma:neartooptimal}. We claim that in that iteration, an $(1+\epsilon)$-approximate median is inserted to  $C$. 

\begin{restatable}{lemma}{lemcorrectnessapproxscheme}\label{lemma:correctnessapproxscheme}
Let $c$ be an optimal restricted $(p,1)$-\mean of $T$. Let $r^{\ast}$ be such that $\cost(T,c) \in [nr^{\ast}, 2nr^{\ast}]$
and let $\gamma^{\ast} = \frac{\epsilon r^{\ast}}{(2m)^{1/p} \sqrt{d}} $. If $\tau_i \in S$ is such that $\dtw_p(\tau_i,c)\leq (2/n) \cdot {\cost(T,c)}$ then 
\begin{enumerate}[label=\roman*)]
\item \label{boundednet}$|\mathbb{G}(\BBB(\tau_i,4r^{\ast}),\gamma^{\ast})| \leq \ell \cdot 2 \cdot \left( \frac{34 r^{\ast}}{\gamma^{\ast} \sqrt{d}}+5 \right)^d$ and 
\item \label{correctanswer} there exists $c' \in \mathbb{G}(\BBB(\tau_i,4r^{\ast}),\gamma^{\ast})^{\leq \ell}$ such that 
$\cost(T,c')\leq (1+\epsilon)\cdot \cost(T,c)$.
\end{enumerate}
\end{restatable}

%\says{dennis}{this proof looks good to me; though we should use consistent notation and I encourage you to use cleveref instead of ref since some Lemmas are called Theorem in your proofs}
%\says{Ioannis}{Switched to cref. Regarding notation. I guess one difference with previous sections is that I use the letter $\pi$ (instead of $c$) for medians. Is that what you mean? Did you notice any other inconsistency? }

The correctness of our algorithm follows by combining the above. 
\begin{restatable}{lemma}{lemoneplusepscorrectness}\label{lemma:oneplusepscorrectness}
    Given a set $T\subset\left( \RR^d\right)^{\leq m}$, $\epsilon>0$, $p\in [1,\infty)$, $\delta\in (0,1)$, \cref{alg:alpha_apprx_median} 
    returns a $(1+\epsilon)$-approximate restricted $(p,1)$-\mean with probability of success $1-\delta$.
\end{restatable}
\begin{proof}
Let $c$ be an optimal restricted $(p,1)$-\mean of $T$. 
 Applying a union bound over the events of \cref{lemma:roughapprox} and \cref{lemma:neartooptimal}, we conclude that with probability at least $1-\delta$, we have
$R\leq  8m^{1/p}\ell^{1/p} \cdot \cost(T,c)$,
   and there exists a $\tau_i \in S$ such that $ \dtw_p(\tau_i,c)<(2/n)\cdot {\cost(T,c)}$. 
   We show correctness assuming that the above two events hold. By \cref{lemma:rightscale} we know that there exists an $r^{\ast}\in I_R$ such that $\cost(T,c) \in [nr^{\ast}, 2nr^{\ast}]$.

We focus on the iteration where $r^{\ast}$ is considered.
Let $\gamma^{\ast}$ be the value of $\gamma$ in that  iteration 
and let $N^{\ast} $ be the set $N$ in that iteration. By \cref{lemma:correctnessapproxscheme} i), $|N^{\ast}|\leq \ell \beta$ and all point sequences of complexity at most $\ell$ defined by points in  ${N^{\ast}}$ will be considered as possible solutions. By \cref{lemma:correctnessapproxscheme} ii), there is a point sequence in  $\left(N^{\ast}\right)^{\leq \ell}$ which is a $(1+\epsilon)$-approximate solution.
\end{proof}

Finally, we bound the running time of \cref{alg:alpha_apprx_median}.
\begin{theorem}
\label{theorem:onepluseps}
     Given a set $T\subset\left( \RR^d\right)^{\leq m}$ of $n$ point sequences, $\epsilon\in (0,m^{1/p}]$, $p\in [1,\infty)$, $\delta\in (0,1)$, \cref{alg:alpha_apprx_median} returns a $(1+\epsilon)$-approximate restricted $(p,1)$-\mean with probability of success $1-\delta$ and has running time in   %\[O\left( \left(  m^4\cdot   + n m  \cdot 2^{\ell} \cdot \ell^{\ell} \cdot \mathsf{g}\left(\frac{m^{1/p}}{\epsilon}  \right)^{d\ell} \cdot \frac{\log(m\ell)}{p} \right)\cdot d\ell \log\left(\frac{1}{\delta}\right)  \right). \]
     $O\left( \left(  m^4   + n m   \cdot \left(\frac{m^{1/p}}{\epsilon}  \right)^{d\ell} \cdot \frac{\log(m)}{p} \right)\cdot  \log\left(\frac{1}{\delta}\right)  \right)$.
\end{theorem}

\begin{proof}
Correctness follows from \cref{lemma:oneplusepscorrectness}. It remains to bound the running time. 
 For each one of the point sequences in $S$, we compute its  $(2,\ell)$-simplification in $O(d m^4\ell)$ time  using \cref{theorem:simplification} and its median cost in $O(dnm\ell)$ time. Hence, the total time needed to compute $\mathcal{R}$ and then $R$ is $O((m^4+nm) \cdot d \ell  \log(1/\delta))$. The set $I_R$ has cardinality $|I_R|= O\left( \log (m\ell)/p\right)$. For each value $r\in I_R$, we add at most 
$     \sum_{i=1}^{\ell} |N|^{i}\cdot |S|\leq \ell
     \left(\ell \cdot 2\cdot\left(  68 m^{1/p}\epsilon^{-1} + 5\right)^d \right)^{\ell}\cdot|S|$
 candidates. For each candidate point sequence in $C$, we compute the cost in time $O(dnm\ell)$. Since $d$ and $\ell$ are considered constants, the total running time is
 %\[O\left( \left(  m^4\cdot   + n m  \cdot 2^{\ell} \cdot \ell^{\ell} \cdot \mathsf{g}\left(\frac{m^{1/p}}{\epsilon}  \right)^{d\ell} \cdot \frac{\log(m\ell)}{p} \right)\cdot d\ell \log\left(\frac{1}{\delta}\right)  \right).\]
  $O\left( \left(  m^4  + n m   \cdot \left(m^{1/p}\epsilon^{-1}  \right)^{d\ell} \cdot \log(m)/p \right)\cdot  \log\left(1/\delta\right)  \right)$.
\end{proof}

\section{Application to Clustering}\label{sec:clustering}

We can apply the results of \cref{section:constantfactor,section:simpltrineq} to the problem of clustering of point sequences. The details can be found in Appendix~\ref{appendix:clustering}.  
\begin{definition}[$(k,\ell,p,q)$-clustering]
    The $(k,\ell,p,q)$-clustering problem is defined as follows, where $k\in \mathbb{N}$, $\ell \in \mathbb{N}_{>1}$ and $p, q \in [1, \infty)$ are fixed (constant) parameters of the problem: given a set $T = \{ \tau_1, \dots, \tau_n \} \subseteq X^{\leq m}$ of point sequences, compute a set $C \subseteq X^{\leq \ell}$ of $k$ point sequences, such that $ \cost_p^q(T,C) = \sum_{i=1}^n \min_{c \in C} \dtw_p(c, \tau_i)^q$ is minimal.
\end{definition}

Solving an instance of the  $(k,\ell,p,q)$-clustering problem is equivalent to solving an instance of the $k$-median problem, where the distance between any center $c$ and any other element $x$ is measured by $\dtw_p(x,c)^q$.  
To solve the $k$-median problem, we  apply a generic framework, as described in  \cite{BDS20}, which reduces the clustering problem to a problem of computing a set that contains medians for sufficiently large subsets of the input, with good probability.   
To solve this latter problem, we adapt our algorithms from Sections \ref{section:constantfactor} and \ref{section:simpltrineq}. \cref{alg:1_median_2_3} can be easily modified to return the set of candidates, instead of returning the best among them. We show that setting parameters appropriately yields an algorithm that satisfies the properties required by the above-mentioned framework and 
leads to a randomized algorithm for the $(k,\ell,p,p)$-clustering problem  with approximation factor in $O(2^p)$, probability of success $1-\delta$ and running time in $O\left(({2^p k m}  \ln\left({\ell}/{\delta}\right))^{\ell(k+2)}
\cdot  n  m   \right)$, assuming that the time needed to compute the distance between two points  is constant (see \cref{theorem:kmedian1} in Appendix~\ref{appendix:clustering}). Similarly, the random sampling method implied by \cref{lem:sampleonesequence} can be used to produce a sufficiently large sample of candidates, which leads to a randomized algorithm for the $(k,\ell,p,1)$-clustering problem with approximation factor in $O(m^{1/p}\ell^{1/p})$, probability of success $1-\delta$ and running time in $O\left((k \log(1/\delta))^{k+2}
\cdot  n  m  +
m^4(k \log(1/\delta))^{k+2} \right)$, assuming again constant time for distance computations of points (see \cref{theorem:kmedian2} in Appendix~\ref{appendix:clustering}).

\section{Conclusions}
We have studied \mean problems for point sequences under the $p$-DTW distance and devised exact and approximation algorithms for several relevant problem variants where the complexity of the \mean is restricted by a parameter $\ell$. Our exact algorithm runs in polynomial time for constant $\ell$ and $d$. 
The running times of our approximation algorithms depend only linearly on the number of input sequences. The dependency on the length of the input sequences, however, is high;  the dependency on the parameter $\ell$ is even exponential. 
We hope that the algorithmic techniques developed in this paper will inspire further work on the topic. %It would be interesting, for example, to extend our techniques to other metric spaces or to use them as building blocks for future algorithmic tools. 
In particular, we think the weak triangle inequality and the simplification algorithm could be of great use. 
In contrast, a proof of hardness of approximation for the central problem studied in this paper is not in sight. We leave this as an open problem.
%In contrast, an extension of our lower bound to hardness of approximation for the constrained problems is not in sight, leaving this an open problem. 
%\todo{K: Maybe I'm just not getting it. But haven't we shown that the constrained problems are basically tractable (for constant $\ell$)? So why looks for hardness of approximation there?} 
%It is also an interesting open problem whether there exists an $\mathrm{FPT}$-algorithm for the restricted problem that is parameterized by the length of the mean sequence $\ell$. \todo{DR: It is also interesting whether the unrestricted problem is contained in poly-APX or exp-APX or is NPO-complete. (We now only ruled out that it is contained in log-APX (unless $NP\subseteq QP$) and in APX (unless $P = NP$).}

\bibliography{bibliography}

\appendix
\section{Appendix}
\label{sec:appendix}

\subsection{Missing Proof of Section \ref{sec:exact_computation}\label{appendix:exact}}

\exactmeanalgorithm*
Before we prove the theorem we note that the employed cylindrical algebraic decomposition algorithm only works for polynomials with real algebraic coefficients. Therefore, here we assume that the input point sequences have rational coordinates. This is indeed a realistic assumption since physical computers are not capable of storing arbitrary real numbers.
\begin{proof}[Proof of \cref{theorem:exactcomputation}]

To simplify our exposition, we restrict ourselves to \mean{s} of complexity exactly $\ell' \in [\ell]$, i.e.\ in the rest of the proof we describe an algorithm for computing an optimal \mean of complexity exactly $\ell'$. The complete algorithm consists of iteratively  computing the optimal \mean of complexity $\ell'$, for each $\ell' \in [\ell]$. %This only increases the complexity by a factor of $O(\ell)$ which is assumed to be constant.  

    For each $\tau = (\tau_1, \dots, \tau_{\lvert \tau \rvert}) \in T$ and all $W_1, W_2 \in \mathcal{W}_{{\ell'}, \lvert \tau \rvert}$ we define for $c = (c_1, \dots, c_{\ell'}) \in \left(\mathbb{R}^d\right)^{\ell'}$ the polynomial function $$ P_{\tau, W_1, W_2}(c) = \left(\sum_{(i,j) \in W_1} \lVert c_i - \tau_j \rVert^2 \right) - \left( \sum_{(i,j) \in W_2} \lVert c_i - \tau_j \rVert^2 \right).$$
    
    Clearly, iff $W_1$ yields a smaller distance between $c$ and $\tau$ than $W_2$, then $P_{\tau, W_1, W_2}(c) < 0$ and iff $W_2$ yields a smaller distance between $c$ and $\tau$ than $W_2$, then $P_{\tau, W_1, W_2}(c) > 0$. Iff $P_{\tau, W_1, W_2}(c) = 0$, both yield the same distance.
    
    Let $F = \{ P_{\tau, W_1, W_2} \mid \tau \in T, W_1, W_2 \in \mathcal{W}_{{\ell'}, \lvert \tau \rvert} \}$ be the set of these polynomials. The central observation is that if all functions in $F$ have the same sign for any $c_1, c_2 \in \left(\mathbb{R}^d\right)^{\ell'}$, then $c_1$ and $c_2$ have the same optimal $2$-warpings to the point sequences in $T$. To see this, for each $\tau \in T$ let $W_{\tau} \in \mathcal{W}_{{\ell'}, \lvert \tau \rvert}$ be an optimal $2$-warping between $c_1$ and $\tau$. Clearly, $P_{\tau, W_{\tau}, W}(c_1) \leq 0$ for all $\tau \in T$ and $W \in \mathcal{W}_{{\ell'}, \lvert \tau \rvert}$. Now, if all functions in $F$ have the same sign for $c_1$ and $c_2$ it must be that $P_{\tau, W_{\tau}, W}(c_2) \leq 0$ for all $\tau \in T$ and $W \in \mathcal{W}_{{\ell'}, \lvert \tau \rvert}$. Thus, $W_{\tau}$ is an optimal $2$-warping between $c_2$ and $\tau$ for each $\tau \in T$.
    
    Now, we compute an arrangement of the zero sets of the polynomials in $F$ (cf. \cite{DBLP:books/daglib/0018467}), i.e., a partition of $\left(\mathbb{R}^d\right)^{\ell'}$ into regions where all functions in $F$ have the same sign. For this purpose we use the cylindrical algebraic decomposition algorithm \cite{10.1007/3-540-07407-4_17}, which also yields a sample from each face of the arrangement and has running time $O(\lvert F \rvert^{f({\ell'} \cdot d)})$ for some function $f \in 2^{O(d {\ell'})}$. For each sample $c$ from some face of the arrangement we first compute the optimal $2$-warpings between $c$ and the input point sequences $\tau \in T$ in time $O(nm)$. Second we compute all sections $S_j(c)$ of $c$ and stores the point sequence $c^\prime = (c^\prime_1, \dots, c^\prime_{\ell'})$ consisting of the means $c^\prime_j = \frac{1}{\lvert S_j(c) \rvert} \sum_{v \in S_j(c)} v$ for $j = 1, \dots, {\ell'}$. This takes time $O(nm)$.
    
    At some point, we obtain a sample from the face containing the optimal $({\ell'}, 2, 2)$-\mean $c^\ast = (c^\ast_1, \dots, c^\ast_{\ell'})$ (where $c^\ast_j$ is the mean of $S_j(c^\ast)$ for each $j \in [{\ell'}]$), which we return when we finally return the point sequence $c^\prime$ that minimizes the objective function. This takes time $O(n m A)$, where $A$ is the number of cells in the arrangement.
    
    To conclude the proof, note that for each $\tau \in T$ we have that $\lvert \mathcal{W}_{{\ell'}, \lvert \tau \rvert} \rvert \leq m^{2 {\ell'}}$, thus $\lvert F \rvert \leq n m^{4 {\ell'}}$. Hence, $A \leq \left(\frac{100 n m^{4 {\ell'}}}{{\ell'} d}\right)^{{\ell'} d}$ by \cite[Theorem 6.2.1]{DBLP:books/daglib/0018467}.

    As we have already mentioned, we iteratively run the above algorithm to compute \mean{s} of complexity  $\ell'$, for each $\ell' \in [\ell]$, in order to find an optimal restricted $2$-\mean. Each iteration runs in $(nm)^{2^{O(d{\ell'})}}\leq (nm)^{2^{O(d{\ell})}}$. Since $\ell$ is constant, the running time is in $(nm)^{2^{O(d\ell)}}$.
\end{proof}

\subsection{Missing Proofs of Section~\ref{section:constantfactor}}\label{appendix:constant}
\lemmadiscreteapprox*
\begin{proof}
    Let $c = (c_1, \dots, c_{\ell'}) \in X^{\ell'}$ be an optimal restricted $p$-\mean for $T$ and for $j \in [\ell']$ let $S_j = \{s_{j,1}, \dots, s_{j,n_j} \} = S^p_j(c)$ for brevity. Define $\Delta(S_j) = \sum\limits_{v \in S_j} \rho(c_j, v)^p$. We immediately have $\cost(T,c) = \sum_{j=1}^{\ell'} \Delta(S_j)$. Now, for $j \in [\ell']$, let $\pi_j$ be a permutation of the index set $[n_j]$, such that $$ \rho(c_j, s_{j, \pi_j^{-1}(1)})^p \leq \dots \leq \rho(c_j, s_{j, \pi_j^{-1}(n_j)})^p.$$
    
    Let $\epsilon^\prime = \frac{\epsilon}{2^{p-1}+\epsilon}$. For the sake of simplicity, we assume that $\epsilon^\prime n$ is integral. Further, for $j \in [\ell']$, we define $C_j = \{ s_{j, \pi_j^{-1}(1)}, \dots, s_{j, \pi_j^{-1}(\epsilon^\prime n)} \}$. We have that $\rho(c_j, s_{j, \pi_j^{-1}(\epsilon^\prime n)})^p \leq \frac{\Delta(S_j)}{\lvert S_j \rvert - (\epsilon^\prime n - 1)}$ by the fact that $\rho(c_j, s_{j, \pi_j^{-1}(\epsilon^\prime n)})^p$ is maximal, iff $\rho(c_j, s^\prime)^p = 0$ for each $s^\prime \in C_j \setminus \{s_{j, \pi_j^{-1}(\epsilon^\prime n)}\}$ and $\rho(c_j, s^\prime)^p = \rho(c_j, s_{j, \pi_j^{-1}(\epsilon^\prime n)})^p$ for each $s^\prime \in S_j \setminus C_j$. For $j \in [\ell']$, we now define $$B_j = \{ x \in P \mid \rho(c_j, x)^p \leq \rho(c_j, s_{j, \pi_j^{-1}(\epsilon^\prime n)})^p \}$$ and by definition we have $\rho(c_j, x)^p \leq \frac{\Delta(S_j)}{\lvert S_j \rvert - \epsilon^\prime n + 1} \leq \frac{\Delta(S_j)}{\lvert S_j \rvert - \epsilon^\prime n}$ for each $x \in B_j$ and $j \in [\ell']$. Then let $c^\prime = (c^\prime_1, \dots, c^\prime_{\ell'})$ be a point sequence with $c^\prime_j \in B_j$ for each $j \in [\ell']$. We bound its cost:
    \begin{align*}
        \cost(T,c^\prime) & = \sum_{j=1}^{\ell'} \sum_{v \in S_j} \rho(c^\prime_j, v)^p  \leq \sum_{j=1}^{\ell'} \sum_{v \in S_j} (\rho(c_j, v) + \rho(c_j, c^\prime_j))^p \\
        & \leq \sum_{j=1}^{\ell'} \sum_{v \in S_j} 2^{p-1} (\rho(c_j, v)^p + \rho(c_j, c^\prime_j)^p) \leq 2^{p-1} \sum_{j=1}^{\ell'} \sum_{v \in S_j}  \left(\rho(c_j, v)^p + \frac{\Delta(S_j)}{\lvert S_j \rvert - \epsilon^\prime n}\right) \\
        & \leq 2^{p-1} \cost(T,c) + 2^{p-1} \sum_{j=1}^{\ell'} \sum_{v \in S_j} \frac{\Delta(S_j)}{(1-\epsilon^\prime) \vert S_j \rvert} = \left(2^{p-1}+\frac{2^{p-1}}{1-\epsilon^\prime}\right) \cost(T,c) \\
        & = (2^{p} + \epsilon) \cost(T,c).
    \end{align*}
    The first inequality follows from the triangle-inequality and the last inequality holds, because a vertex from each $\tau_i \in T$ must be warped to each $c_j \in c$, thus $\lvert S_j \rvert \geq n$ for each $j \in [\ell']$.
\end{proof}

\thmmedianconstantfactor*
\begin{proof}%[Proof of \cref{thm:1_median_2_3}]
    For the given $\epsilon$, let $\epsilon^\prime = \frac{\epsilon}{2^{p-1} + \epsilon}$ and let $B_1, \dots, B_{\ell'}$, $\ell'\leq \ell$, be the balls guaranteed by \cref{lem:discrete_approx}. Recall that each ball has size at least $\epsilon^\prime n$. For each $i \in [\ell']$ and $s \in S$ we have $\Pr[s \not\in B_i] \leq (1-\frac{\epsilon^\prime n}{\lvert P \rvert}) \leq (1-\frac{\epsilon^\prime n}{n m}) = (1-\frac{\epsilon^\prime}{m}) \leq \exp(-\epsilon^\prime/m)$.
    
    By independence, for each $i \in [\ell']$ we have $\Pr[B_i \cap S = \emptyset] \leq \exp(-\epsilon^\prime/m)^{\left\lceil \frac{m(\ln(\ell) - \ln(\delta))}{\epsilon^\prime} \right\rceil} \leq \delta/\ell$. Using a union bound we conclude that with probability at least $1-\delta$, $S$ contains at least one element of $B_i$, for each $i \in [\ell']$, and thus \cref{alg:1_median_2_3} returns a $(2^{p}+\epsilon)$-approximate restricted $p$-\mean for $T$ with probability at least $1-\delta$ by \cref{lem:discrete_approx}.
    \\
    
    The running time of the algorithm is dominated by computing the cost of all point sequences of complexity at most $\ell$ over $S$. Since $\lvert S^{\leq \ell} \rvert$ is in $O\left(\frac{\ln(1/\delta)^\ell m^\ell}{(\epsilon^\prime)^\ell}\right) = O\left(\ln(1/\delta)^\ell m^\ell \frac{(2^{p-1}+\epsilon)^\ell}{\epsilon^\ell} \right)$ and every distance can be computed in time $O(m)$, this takes time $O\left(\ln(1/\delta)^\ell m^{\ell+1} n \frac{(2^{p-1}+\epsilon)^\ell}{\epsilon^\ell}\right)$.
\end{proof}

\thmmedianconstantdetermcor*
\begin{proof}%[Proof of \cref{thm:det_mean_correctness}]
    By \cref{lem:discrete_approx}, for any $\epsilon \in \left(0, \infty\right)$ there exist balls $B_1, \dots, B_{\ell'} \subseteq P$, $\ell'\leq \ell$ of cardinality at least $\epsilon^\prime n$ each, such that any point sequence $c^\prime = (c^\prime_1, \dots, c^\prime_{\ell'})$, with $c^\prime_i \in B_i$ for each $i \in [\ell']$, is a $(2^{p}+\epsilon)$-approximate restricted $p$-\mean for $T$, where $\epsilon^\prime = \frac{\epsilon}{2^{p-1} + \epsilon}$. Since we compute an $(\epsilon^\prime/m)$-net of $P$ and $\lvert P \rvert \leq nm$, $S$ contains at least one point from each of $B_1,\ldots,B_{\ell'}$ by \cref{def:epsilon_net}. Hence, $S^{\leq \ell}$ contains a $(2^{p}+\epsilon)$-approximate restricted $p$-\mean for $T$.
\end{proof}

\lemmasubsystemoracle*
\begin{proof}%[Proof of \cref{lem:subspaceoracle}]
    The VC dimension of $(P,\mathcal{B})$ is bounded by $d+1$, see \cite{H11}. For any $Y\subseteq P$, we need to compute the set $ {\mathcal{B}}_{|Y}$ explicitly in time $O(|Y|^{d+2})$. We first apply the standard lifting $\phi\colon (x_1,\ldots,x_d)\mapsto \left(x_1,\ldots,x_d,\sum_{i=1}^d x_i^2\right)$. A point $p\in Y$ belongs to some ball $B \in \mathcal{B}$, with center $c = (c_1,\ldots,c_d) \in \mathbb{R}^d$ and radius $r>0$, if and only if $\phi(p)$ lies below the hyperplane $h_B$, where $h_B$ is the hyperplane defined by the equation $\langle a_B , x \rangle = b_B$, where $a_B=(2c_1,2c_2,\ldots 2c_d,1)$ and $b_B=r^2-\sum_{i=1}^d c_i^2$. Notice that $h_B$ is nonvertical by definition. Then we dualize: for any point $\phi(p)=(y_1,\ldots,y_{d+1})$, $D(\phi(p))=\{(x_1,\ldots,x_{d+1})\in \mathbb{R}^{d+1}\mid  x_{d+1}=\sum_{i=1}^d x_iy_i -y_{d+1}\}$ is a nonvertical hyperplane in $\mathbb{R}^{d+1}$ and for any nonvertical hyperplane $h_B$, $D^{-1}(h_B)$ is a point in $\mathbb{R}^{d+1}$. A standard fact about duality is that a point $\phi(p)$ lies below a hyperplane $h_B$ if and only if the hyperplane $D(\phi(p))$ lies above point $D^{-1}(h_B)$. Finally we construct the arrangement of hyperplanes in the dual space in time $O(|Y|^{d+1})$, using the algorithm in \cite{EOS86}. For each of the at most $O(|Y|^{d+1})$ cells, we return a subset $X\subseteq Y$  corresponding to the hyperplanes lying above. The overall running time is $O(|Y|^{d+2})$. 
\end{proof}

\thmmedianconstantdeterfinal*
\begin{proof}%[Proof of \cref{thm:constantfactorderandomized}]
    The VC dimension of the range space $(P,\mathcal{B})$ is bounded by $d+1$, see \cite{H11}. By \cref{lem:subspaceoracle}, we can use \cref{thm:determnet} to compute an $(\epsilon^\prime/m)$-net $S$ of $(P,\mathcal{B})$, with size $\lvert S \rvert = O\left(\frac{m}{\epsilon^\prime} \log \left( \frac{m}{\epsilon^\prime} \right) \right)$, in time $O\left(nm\left(\frac{m}{\epsilon^\prime} \log \left(\frac{m}{\epsilon^\prime} \right)\right)^{d+1}\right)$. We then compute the $\dtw_p$ distance of any of the $\lvert S^{\leq \ell} \rvert$ candidates with the $n$ input point sequences in time $O\left(\lvert S \rvert^{\ell} \cdot n m \right)$. 
\end{proof}

\subsection{Details of Section~\ref{section:simpltrineq}}
\label{appendix:simpltrineq}
\subsubsection{Minimum-error Simplification}
We present a dynamic programming solution for the problem of computing an approximate simplification. Our algorithm can be seen as a special case of the result of Brill et al.~\cite{BFFJNS18, DBLP:journals/datamine/BrillFFJNS19} for computing a \mean of restricted complexity, but since our statement is different, we include a proof for completeness.%\says{anne}{what exactly is the difference?}\says{anne}{Was this already used in Koen's thesis?}
%\says{ioannis}{Koen also has a dynamic programming solution. Differences: \begin{itemize}    \item Brill et al.: they use $\dtw_2$ (I just noticed that. I thought they were using $\dtw_1$). The restriction on the complexity of the mean is not incorporated in the algorithm. There is only a remark saying "it is possible to ...."    \item Koen's thesis: the statement is only for the exact computation, and it's of the following form: if there exists an exact algorithm for the $1$-median of points with complexity $f(n)$ then there is an algorithm for the $(1,\ell)$-median of sequences with complexity $O(f(mn)...)$. (We implicitly use an approximate algorithm for the $1$-median of points, because we restrict the vertices of the simplification to input points)\end{itemize}It is possible thought, that an algorithm for computing such a simplification appears somewhere else... We can put it in the appendix for completeness.}

\begin{algorithm}[H]
\caption{$2$-Approximate Simplification\label{alg:simpl}}
    \begin{algorithmic}[1]
        \Procedure{$2$-Approximate-Simplification}{$\pi =  (x_1 , \dots, x_{m}) , \ell $, $p$}
            \State Initialize $m \times \ell$ table $D$ with elements in $\mathbb{R}$ 
            \State Initialize $m \times \ell$ table $C$ with elements in $\mathbb{R}^{\leq \ell}$
            %\State $D(0,0)\gets 0$ 
            %\State $C(0,0)\gets (\cdot)$
            \State $P\gets \{x_1,\ldots,x_m \}$
            \For{\textbf{each } $i=1,\ldots,m$}
                \For{\textbf{each } $j=1,\ldots,\ell$}
                 \If{$j=1$}
                \State $x^{\ast} \gets \argmin_{x\in P} \sum_{k=1}^i \rho(x_k, x)^p$
                \State $D(i,j)\gets \sum_{k=1}^i \rho(x_k, x^\ast)^p$
                \State $C(i,j) \gets  (x^{\ast})$
              \Else
                \State $i' \gets \argmin_{i'\leq i}\left(D(i',j-1)+\min_{x\in P} \sum_{k=i'}^i \rho(x_k, x)^p\right)$
                \State $x^{\ast} \gets \argmin_{x\in P} \sum_{k=i'}^i \rho(x_k, x)^p$
                \State $D(i,j)\gets D(i',j-1)+ \sum_{k=i'}^i \rho(x_k, x^\ast)^p$
                \State $C(i,j) \gets C(i',j-1) \oplus (x^{\ast})$
                \EndIf
                \EndFor
            \EndFor
            \State $j^{\ast}\gets \argmin_{j \in [\ell]} D(m,j)$
            \State \Return $C(m,j^{\ast})$
        \EndProcedure
    \end{algorithmic}
\end{algorithm}
\begin{lemma}
\label{lemma:simplificationcorrectness1}
Let $\mathcal{X}=(X,\rho)$ be a metric space. Given as input a point sequence $\pi=(x_1,\ldots,x_m) \in X^m$, \cref{alg:simpl} returns a point sequence from $\{x_1,\ldots,x_m \}^{\leq \ell}$, which minimizes the  $\dtw_p$ distance to $\pi$, among all point sequences in $\{x_1,\ldots,x_m \}^{\leq \ell}$. % in time $O(m^{4}\ell\cdot T_{\rho})$, where $T_{\rho}$ denotes the worst-case running time needed to compute $\rho(\cdot,\cdot)$. 
\end{lemma}
\begin{proof}
    %Let  $P=\{x_1,\ldots,x_m \}$. 
    We show that 
    $C(m,j^{\ast})$ satisfies \[\dtw_p\left(\pi,C(m,j^{\ast})\right)=\min_{\pi'\in P^{\leq \ell}} \dtw_p (\pi,\pi').\]  
   % For each $i\in [m]$, let $\pi_{|i}=(x_1,\ldots,x_i)$.   We  claim  that for any $i\in[m]$, $j\in[\ell]$,  there is a point sequence $\tilde{\pi}\in P^{\leq j}$ such that $\dtw_p(\pi_{|i},\tilde{\pi})=\min_{\pi'\in P^{\leq j}}\dtw_p(\pi,\pi')$, and such that the optimal warping between $\pi_{|i}$ and $\tilde{\pi}$ does not match two vertices of $\tilde{\pi}$ with the same vertex of $\pi_{|i}$. To see this, consider an optimal warping $W \in \mathcal{W}_{m,\ell}$ between $\pi$ and some point sequence $\pi'=(p_1',\ldots,p_{j}')\in P^{\leq j}$. Let $(t_1,t_2)\in W$ and $(t_1,t_2+1)\in W$. If $(t_1-1,t_2)\in W$ and $(t_1+1,t_2+1)\in W$ then removing any one of $(t_1,t_2), (t_1,t_2+1)$ from $W$ yields a new warping with a cost at most equal to the cost of $W$. If  $(t_1-1,t_2)\notin W$, then we can remove $p_{t_2}'$ from $\pi'$.  If  $(t_1+1,t_2+1)\notin W$, then we can remove $p_{t_2+1}'$ from $\pi'$. We conclude that there exists a point sequence $\pi''\in P^{\leq j}$ such that $\dtw_p(\pi_{|i},\pi'')\leq \dtw_p(\pi_{|i},\pi')$, and an optimal warping $W$ between $\pi_{|i}$ and $\pi''$ for which there are no $t_1\in[m], t_2\in[\ell]$ such that both $(t_1,t_2)\in W$ and $(t_1,t_2+1)\in W$. 

  We  claim  that  there is a point sequence $\tilde{\pi}\in P^{\leq \ell}$ such that \[\dtw_p(\pi,\tilde{\pi})=\min_{\pi'\in P^{\leq \ell}}\dtw_p(\pi,\pi'),\] and such that the optimal warping between $\pi$ and $\tilde{\pi}$ does not match two vertices of $\tilde{\pi}$ with the same vertex of $\pi$. To see this, consider an optimal warping $W $ between $\pi$ and some point sequence $\pi'=(p_1',\ldots,p_{j}')\in P^{\leq \ell}$. Let $(t_1,t_2)\in W$ and $(t_1,t_2+1)\in W$. If $(t_1-1,t_2)\in W$ then removing $(t_1,t_2)$ yields a new warping with a cost at most equal to the cost of $W$. Similarly, if $(t_1+1,t_2+1)\in W$ then removing $ (t_1,t_2+1)$ from $W$ yields a new warping with a cost at most equal to the cost of $W$. If  $(t_1-1,t_2)\notin W$, then we can remove $p_{t_2}'$ from $\pi'$.  If  $(t_1+1,t_2+1)\notin W$, then we can remove $p_{t_2+1}'$ from $\pi'$. We conclude that there exists a point sequence $\pi''\in P^{\leq \ell}$ such that $\dtw_p(\pi,\pi'')\leq \dtw_p(\pi,\pi')$, and an optimal warping $W$ between $\pi$ and $\pi''$ for which there are no $t_1\in[m], t_2\in[\ell]$ such that both $(t_1,t_2)\in W$ and $(t_1,t_2+1)\in W$. 

For each $i\in [m]$, let $\pi_{|i}=(x_1,\ldots,x_i)$. 
By construction, each $D(i,j)$ stores the minimum distance between $\pi_{|i}$ and any point sequence $x$ from $P^{j}$, where the distance is attained by a warping that does not match two vertices of $x$ to the same vertex of $\pi$. 
Hence,  $D(m,j^{\ast})$ stores the minimum distance between $\pi$ and any point sequence in $P^{\leq \ell}$, and $C(m,j^{\ast})$ stores a point sequence from  $P^{\leq \ell}$ with distance $D(m,j^{\ast})$ from $\pi$. 
%Hence,  for any $i\in[m], j\in[\ell]$, $D(i,j)$ stores the minimum distance between $\pi_{|i}$ and any point sequence in $P^{\leq j}$, and $C(i,j)$ stores a point sequence from  $P^{\leq j}$ with distance $D(i,j)$ from $\pi_{|i}$. 
\end{proof}
\begin{lemma}
\label{lemma:simplificationcorrectness2}
Let $\mathcal{X}=(X,\rho)$ be a metric space. Given as input a point sequence $\pi \in X^m$, \cref{alg:simpl} returns a $(2,\ell)$-simplification under the $\dtw_p$ distance.%, in time $O(m^{4}\ell\cdot T_{\rho})$, where $T_{\rho}$ denotes the worst-case running time needed to compute $\rho(\cdot,\cdot)$. 
\end{lemma}
\begin{proof}
Let  $P=\{x_1,\ldots,x_m \}$. 
%We first show that \cref{alg:simpl} returns an $\ell$-simplification $\tilde{\pi}$ which minimizes the distance from $\pi$, among all sequences in $P^{\leq \ell}$. We then show that  $\tilde{\pi}$ is a $2$-approximate   minimum-error-$\ell$-simplification of $\pi$.
%For each $i\in[m], j\in[\ell]$, $D(i,j)$ stores the minimum distance between $x_1,\ldots,x_i$ and any sequence in $P^j$. For each $i\in[m], j\in[\ell]$, $C(i,j)$ stores a sequence from  $P^j$ with distance $D(i,j)$ from $x_1,\ldots,x_i$. %Initially,  $D(1,1)\gets 0$, $C(1,1)\gets x_1$. For each $i\in [m],j\in[\ell]$, we update $D$ and C as follows:\[i'=\argmin_{i'\leq i}\left\{D(i',j-1)+\min_{x\in P} \sum_{k=i'}^i \rho(x_k, x)^p\right\}\]
%\[D(i,j)\gets D(i',j-1)+\min_{x\in P} \sum_{k=i'}^i \rho(x_k, x)^p ,\]
%\[ C(i,j) \gets \text{concatenate }C(i',j-1), \left(\argmin_{x\in P} \sum_{k=i'}^i \rho(x_k, x)^p\right).\] 
%We output $C(m,j^{\ast})$, where $j^{\ast}=\argmin_{j \in [\ell]} C(m,j)$.  
By Lemma \ref{lemma:simplificationcorrectness1},  $C(m,j^{\ast})$ is a point sequence in $P^{\leq \ell}$ that minimizes the distance to $\pi$, among all point sequences in $P^{\leq \ell}$. 
We show that $C(m,j^{\ast})$ is a $(2,\ell)$-simplification. Let $\pi^{\ast}=(x_1^{\ast},\ldots,x_{\ell'}^{\ast})$ be a $(1,\ell)$-simplification of $\pi$, 
and let $\tilde{\pi}^{\ast}=(\tilde{x}_1^{\ast},\ldots,\tilde{x}_{\ell'}^{\ast})$, where for each $i\in [\ell']$, 
$\tilde{x}_i^{\ast} := \argmin_{x\in P} \rho(x,x_i^{\ast})$. Let $W^{\ast} \in \mathcal{W}_{m,\ell'}$ be an optimal warping of $\pi$ and  $\pi^{\ast}$. Then,
\begin{align}
    \dtw_{p}^1(\pi,C(m,j^{\ast})) &\leq \dtw_{p}^1(\pi,\tilde{\pi}^{\ast}) \nonumber \\
    &=\min_{W\in \mathcal{W}_{m\ell}} \left(\sum_{(i,j)\in W} \rho(x_i,\tilde{x}_j^{\ast})^p\right)^{1/p} \nonumber \\
    &\leq \left(\sum_{(i,j)\in W^{\ast}} \rho(x_i,\tilde{x}_j^{\ast})^p\right)^{1/p} \nonumber\\
     &\leq \left(\sum_{(i,j)\in W^{\ast}} \left(\rho(x_i,x_j^{\ast})+\rho(x_j^{\ast},\tilde{x}_j^{\ast})\right)^p\right)^{1/p}\label{eq:simplificationcorrectness2:trineq}\\
     &\leq \left(\sum_{(i,j)\in W^{\ast}} 2^p\left(\rho(x_i,x_j^{\ast})\right)^p\right)^{1/p} \nonumber\\
     &= 2 \cdot \dtw_p(\pi,\pi^{\ast}) \nonumber,
\end{align}
where in Step (\ref{eq:simplificationcorrectness2:trineq}) we applied the triangle inequality.

\end{proof}

\thmsimplification*
\begin{proof}
    Correctness of \cref{alg:simpl} follows from Lemma \ref{lemma:simplificationcorrectness2}. It remains to bound the running time of the algorithm. To do so, we consider the operations taking place in the body of the nested loop. 
    For each $i,j$, we iterate over $O(m)$ values for $i'$ and for each value of $i'$ we compute $\min_{x\in P} \sum_{k=i'}^i \rho(x_k, x)^p$ in time $O((i-i') \cdot m)= O(m^2)$. Hence, the total running time is $O(m^4\ell)$.  
\end{proof}

\subsubsection{Weak Triangle Inequality}
\lemtriangleineq*

\begin{proof}
Let $W_{xz}\in \mathcal{W}_{|x|,|z|}$ be the optimal warping between $x$ and $z$. 
Let $W_{xy}\in \mathcal{W}_{|x|,|y|}$ be the optimal warping between $x$ and $y$ and $W_{yz}\in \mathcal{W}_{|y|,|z|}$ be the optimal warping between $y$ and $z$. Let $S_{xz}=\{(i,k,j)\in [|x|]\times [|y|]\times[|z|]\mid (i,k)\in W_{xy} \text{ and }(k,j)\in W_{yz} \}$ and 
$W_{xz}'=\{ (i,j) \in [|x|]\times[|z|]\mid  \exists k~ (i,k,j)\in S_{xz}\}$ . Then, 
\begin{align*}
{\dtw}_p(x,z)& = \left(\sum_{(i,j)\in W_{xz}} \rho(x_i,z_j)^p \right)^{1/p}\\ &\leq  \left(\sum_{(i,j)\in W_{xz}'} \rho(x_i,z_j)^{p} \right)^{1/p}\\
&\leq \left(\sum_{(i,k,j)\in S_{xz}} \left( \rho(x_i, y_k)+\rho(y_k,z_j) \right)^p \right)^{1/p} \\ &\leq  \left( \sum_{(i,k,j)\in S_{xz}} \rho(x_i, y_k)^p \right)^{1/p}+ \left(\sum_{(i,k,j)\in S_{xz}}\rho(y_k,z_j)^p \right)^{1/p}   \\ &\leq  m_1^{1/p} \cdot {\dtw}_p(x,y)+ m_1^{1/p}  \cdot
 {\dtw}_p(y,z),
\end{align*}
where the second inequality holds by the triangle inequality and the third inequality holds by Minkowski's inequality. 
\end{proof}

% \lemsampleonesequence*
% \begin{proof}
%  Let $c$ be an optimal $(p,1)$-\mean of $T$ with cost $\mathrm{OPT}_{\ell}$. 
% Then,
% \begin{align}
% \pexpected_{\pi}[\mathrm{cost}(T,\tilde{\pi})]&=
% \pexpected_{\pi}\hspace{-0.2em}\Biggl[\sum_{i=1}^n \dtw_p(\tau_i,\tilde{\pi})\Biggr] \leq \pexpected_{\pi}\hspace{-0.2em}\Biggl[m^{\frac{1}{p}} \sum_{i=1}^n \left( \dtw_p(\tau_i,c) +\dtw_p(c,\tilde{\pi}) \right)\Biggr] \label{eq:trineq1}\\
% &=m^{1/p}\cdot (\mathrm{OPT}_{\ell}+n\cdot \pexpected_{\pi}[\dtw_p(c,\tilde\pi)]) \nonumber\\
% &\leq m^{1/p}\cdot (\mathrm{OPT}_{\ell}+n\cdot \ell^{1/p} \cdot \pexpected_{\pi}[\dtw_p(c,\pi) + \dtw_p(\pi,\tilde\pi)]) \label{eq:trineq2}\\ 
% &\leq  m^{1/p}\cdot (\mathrm{OPT}_{\ell}+(1+\alpha)\cdot n\cdot \ell^{1/p} \cdot \pexpected_{\pi}[\dtw_p(c,\pi)]) \nonumber\\
% &=  m^{1/p}\cdot \left(\mathrm{OPT}_{\ell}+(1+\alpha)\cdot n\cdot \ell^{1/p} \cdot \sum_{\pi \in T} \dtw_p(c,\pi)\cdot \frac{1}{n}\right)\nonumber  \\ 
% &=  m^{1/p}\cdot (\mathrm{OPT}_{\ell}+(1+\alpha)\cdot \ell^{1/p} \cdot \mathrm{OPT}_{\ell}) \leq (2+\alpha) m^{1/p} \ell^{1/p} \cdot \mathrm{OPT}_{\ell},\nonumber
% \end{align}
% where in Step (\ref{eq:trineq1}) and Step (\ref{eq:trineq2}) we applied \cref{lem:asymtriangineq2}.
% \end{proof}

\subsection{Missing Proofs of Section~\ref{section:onepluseps}}
\label{appendix:onepluseps}

\lemroughapprox*
\begin{proof}
For any point sequence $\tau_i\in T$, let $\tau_i'$ be a  $(2,\ell)$-simplification. 
Let $\tau_j$ be a randomly sampled point sequence from $T$.  By \cref{lem:sampleonesequence}, 
\[\pexpected \left[\cost(T,\tau_j')\right]\leq 4m^{1/p}\ell^{1/p} \cdot \cost(T,c).\]
By Markov's inequality, 
$
\Pr\left[\cost(T,\tau_j') \geq 8m^{1/p}\ell^{1/p} \cdot \cost(T,c)\right]\leq \frac{1}{2}$. 
Hence, the probability that $R\geq 8m^{1/p}\ell^{1/p} \cdot \cost(T,c)$ is equal to
\[
\Pr\left[ \forall \tau_i \in S: \cost(T,\tau_i') \geq 8m^{1/p}\ell^{1/p} \cdot \cost(T,c) \right] \leq \frac{1}{2^{|S|}}\leq \frac{\delta}{2}. 
\]
\end{proof}

\lemneartooptimal*
\begin{proof}
Let $\tau_i$ be a randomly sampled point sequence from $T$:
\[
\pexpected_{\tau_i}\left[ \dtw_p(\tau_i, c)\right] = \sum_{i=1}^n \dtw_p (\tau_i, c) \cdot \frac{1}{n}= \frac{\cost(T,c)}{n}.
\]
By Markov's inequality, 
$
\Pr \left[ \dtw_p(\tau_i,c) > 2\cdot \frac{\cost(T,c)}{n}\right]\leq \frac{1}{2}$. 
Hence, 
\[
\Pr \left[ \forall \tau_i \in S:\dtw_p(\tau_i,c) > 2\cdot \frac{\cost(T,c)}{n}\right]\leq \frac{1}{2^{|S|}}\leq \frac{\delta}{2}.
\]
\end{proof}

\lemrightscale*
\begin{proof}
%Obviously $ \cost(T,c)\leq R$,\todo{check this}and 
Since $R$ is the cost of a curve of complexity at most $\ell$, we have that $\cost(T,c)\leq R$. 
By assumption, $\cost(T,c) \geq  R/(8 m^{1/p} \ell^{1/p})$. 
By the definition of $I_R$, there exists  $j\geq 0$ such that 
\[2^{-(j+1)} \cdot \frac{ R}{n} \leq \frac{\cost(T,c)}{n} \leq 2^{-j} \cdot \frac{ R}{n}.\]
Hence, the lemma is true for $r= 2^{-(j+1)} \cdot \frac{ R}{n}$.
\end{proof}

\lemvolumebound*
\begin{proof}

We use Binet's second expression \cite{WW96} for the Gamma function $\ln \Gamma(z)$:
$$ \ln \Gamma(z) = z \ln(z) - z + \frac{1}{2} \ln\left(\frac{2 \pi}{z}\right) + \int_0^\infty \frac{2 \arctan\left(\frac{t}{z}\right)}{e^{2\pi t} - 1} \mathrm{d}t . $$
Since $\arctan(x) \geq 0$ for $x \geq 0$ and $e^{2 \pi x} - 1 \geq 0$ for $x \geq 0$, we have the following inequality:
\begin{align}
    & \ln \Gamma(z) \geq z \ln(z) - z + \frac{1}{2} \ln\left(\frac{2 \pi}{z}\right) \nonumber \\
    \iff & \ln \Gamma(z) \geq \ln(z^z) - \ln(e^z) + \ln\left(\sqrt{\frac{2 \pi}{z}}\right) \nonumber \\
    \iff & \Gamma(z) \geq z^z e^{-z} \sqrt{\frac{2 \pi}{z}}  \nonumber \\ 
    \iff & \Gamma(z) \geq  \sqrt{2 \pi} z^{z - \frac{1}{2}} e^{-z}. \label{ineq:gamma}
\end{align}
We apply a standard volumetric argument to upper bound $\left|\mathbb{G}\left( B(x,8 r),\gamma\right)\right|$. 
\begin{align}
    \left|\mathbb{G}\left( B(x,8 r),\gamma\right)\right| \leq 
    \frac{\mathrm{vol}(B(x,8 r+\gamma\sqrt{d}))}{{\gamma}^{d}}&=
    \frac{\pi^{d/2}}{\Gamma(\frac{d}{2}+1)}\cdot \frac{(8 r+\gamma\sqrt{d})^d}{{\gamma}^d} \nonumber\\
    &\leq \frac{\pi^{d/2} \e ^{d/2+1}}{\sqrt{2\pi}\left(\frac{d}{2}+1 \right)^{d/2+1/2}}\cdot \frac{(8 r+\gamma\sqrt{d})^d}{{\gamma}^d}  \label{ineq:gammaapplied}\\
    &\leq \frac{2^{d/2+1/2}\pi^{d/2} \e ^{d/2+1}}{\sqrt{2\pi}\cdot {d}^{d/2+1/2}}\cdot \frac{(8 r+\gamma\sqrt{d})^d}{{\gamma}^d} \nonumber  \\
    & \leq \frac{ \e \cdot (4.2)^d}{\sqrt{\pi}\cdot {d}^{d/2}}\cdot \frac{(8 r+\gamma\sqrt{d})^d}{{\gamma}^d} \nonumber \\
    & \leq 
    2 \cdot \left( \frac{34 r}{\gamma \sqrt{d}}+5 \right)^d, \nonumber
    \end{align}
where in (\ref{ineq:gammaapplied}) we used (\ref{ineq:gamma}). 
\end{proof}
\lemcorrectnessapproxscheme*
\begin{proof}
Let $c = (c_1,\ldots,c_{\ell'})$, where $\ell'\leq \ell$. 
To prove \ref{boundednet}, notice that $\dtw_p(\tau_i,c) \leq 4 r^{\ast}$, which implies that for any vertex $\tau_{i,j}$ of $\tau_i$, there exists a vertex $c_z$ of $c$ such that $\tau_{i,j} \in B(c_z, 4r^{\ast}) $. By the triangle inequality $B(\tau_{i,j},4r^{\ast}) \subseteq B(c_z, 8r^{\ast}) $. Hence,
\begin{align*}
\BBB(\tau_i,4r^{\ast})\subseteq \bigcup_{z=1}^{\ell} B(c_z,8 r^{\ast}) \implies 
|\mathbb{G}(\BBB(\tau_i,4r^{\ast}),\gamma^{\ast})| &\leq   \left|\mathbb{G}\left(\bigcup_{z=1}^{\ell'} B(c_z,8 r^{\ast}),\gamma^{\ast}\right)\right|\\&\leq 
\sum_{z=1}^{\ell}\left|\mathbb{G}\left( B(c_z,8 r^{\ast}),\gamma^{\ast}\right)\right|.
\end{align*}
By \cref{lemma:volumetricbound}, we obtain 
\[
|\mathbb{G}(\BBB(\tau_i,4r^{\ast}),\gamma^{\ast})|
\leq 
\ell\cdot
2 \cdot \left( \frac{34 r^{\ast}}{\gamma^{\ast} \sqrt{d}}+5 \right)^d.
\]

To prove \ref{correctanswer}, notice that 
 all vertices of $c$ are contained in $\BBB(\tau_i,4r^{\ast})$. Hence, for each point $c_z$ there exists a grid point  $\tilde{c}_z \in \mathbb{G}(\BBB(\tau_i,r^{\ast}),\gamma^{\ast})$ such that $\|c_z-\tilde{c}_z \| \leq \gamma^{\ast} \sqrt{d}$. We will show that the point sequence  $\tilde{c}=(\tilde{c}_1,\ldots,\tilde{c}_{\ell'})$ is a $(1+\epsilon)$-approximation. For each $i\in [n]$, $W_i^{\ast}$ denotes the optimal warping of $\tau_i$ with $c$.
    \begin{align*}
         \cost(T,\tilde{c}) &=
         \sum_{i=1}^n \dtw_p(\tau_i , \tilde{c}) \\
         &= \sum_{i=1}^{n} \min_{W\in \mathcal{W}_{\lvert \tau_i \rvert,\ell}} \left(\sum_{(k,j)\in W} \|\tau_{i,k}-\tilde{c}_j\|^p\right)^{1/p}\\  &\leq \sum_{i=1}^n \left(\sum_{(k,j)\in W_{i}^{\ast}} \|\tau_{i,k}-\tilde{c}_j\|^p\right)^{1/p}\\
         &\leq \sum_{i=1}^n          \left(\sum_{(k,j)\in W_{i}^{\ast}} (\|\tau_{i,k}-{c}_j\|+\| c_j-\tilde{c}_j\|)^p\right)^{1/p}\\
         &\leq \sum_{i=1}^n \left( \left(\sum_{(k,j)\in W_{i}^{\ast}} \|\tau_{i,k}-{c}_j\|^p \right)^{1/p}+\left(\sum_{(k,j)\in W_{i}^{\ast}} \| c_j-\tilde{c}_j\|^p\right)^{1/p} \right)\\
     &\leq \sum_{i=1}^n \left(\dtw_p(\tau_i,c) + 
    |W_i^{\ast}|^{1/p}\cdot \gamma^{\ast}\sqrt{d} \right)\\
   % & \leq \sum_{\tau_i\in T} \left(\dtw_p(\tau_i,c) +     \frac{3\epsilon}{2^q\cdot n\cdot } \right)^q\\
    %&= \sum_{\tau_i\in T} \left(\dtw_p^q(\tau_i,c) + \sum_{k=1}^{q} {{q}\choose{k}} \dtw_p^{q-k}(\tau_i,c)\cdot \zeta^k\right)
    &\leq \sum_{i=1}^n \left(\dtw_p(\tau_i,c) + 
    \frac{ \cost(T,c)\cdot \epsilon}{n} \right)\\
    &=(1+\epsilon)\cdot\cost(T,c),
    \end{align*}
    where the second inequality follows from the triangle inequality, and the third inequality follows from Minkowski's inequality. We also make use of the fact that $|W_i^{\ast}|\leq 2m$.
\end{proof}

% \lemoneplusepscorrectness*
% \begin{proof}
% Let $c$ be an optimal restricted $(p,1)$-\mean of $T$. 
%  Applying a union bound over the events of \cref{lemma:roughapprox} and \cref{lemma:neartooptimal}, we conclude that with probability at least $1-\delta$, we have
% $R\leq  8m^{1/p}\ell^{1/p} \cdot \cost(T,c)$,
%   and there exists a $\tau_i \in S$ such that $ \dtw_p(\tau_i,c)<(2/n)\cdot {\cost(T,c)}$. 
%   We show correctness assuming that the above two events hold. By \cref{lemma:rightscale} we know that there exists an $r^{\ast}\in I_R$ such that $\cost(T,c) \in [nr^{\ast}, 2nr^{\ast}]$.

% We focus on the iteration where $r^{\ast}$ is considered.
% Let $\gamma^{\ast}$ be the value of $\gamma$ in that  iteration 
% and let $N^{\ast} $ be the set $N$ in that iteration. By \cref{lemma:correctnessapproxscheme} i), $|N^{\ast}|\leq \ell \beta$ and all point sequences of complexity at most $\ell$ defined by points in  ${N^{\ast}}$ will be considered as possible solutions. Finally, by \cref{lemma:correctnessapproxscheme} ii), there is a point sequence in  $\left(N^{\ast}\right)^{\leq \ell}$ which is a $(1+\epsilon)$-approximate solution.
% \end{proof}

\subsection{Applications to Clustering}
\label{appendix:clustering}
In this section we apply the results of \cref{section:constantfactor,section:simpltrineq} to the problem of clustering of point sequences.  

%\begin{definition}[$(k,\ell,p,q)$-clustering]
 %   The $(k,\ell,p,q)$-clustering problem is defined as follows, where $k\in \mathbb{N}$, $\ell \in \mathbb{N}_{>1}$ and $p, q \in [1, \infty)$ are fixed (constant) parameters of the problem: given a set $T = \{ \tau_1, \dots, \tau_n \} \subseteq X^{\leq m}$ of point sequences, compute a set $C \subseteq X^\leq \ell$ of $k$ point sequences, such that $ \cost_p^q(T,C) = \sum_{i=1}^n \min_{c \in C} \dtw_p(c, \tau_i)^q$ is minimal.
%\end{definition}

Solving a $(k,\ell,p,q)$-clustering problem is equivalent to solving a $k$-median problem, where the distance between any center $c$ and any other element $x$ is measured by $\dtw_p(x,c)^q$.  
To solve the $k$-medians problem, one can apply the following two  theorems, which are proven in \cite{BDS20}, and appear slightly rephrased here. The two theorems provide sufficient conditions for a solution to the  $k$-medians problem. We specialize the  statements to our case of interest, the $\dtw_p$ distances, raised to the power of $q$.  
%\todo{should we cite \cite{ABS10}?}
\begin{theorem}[Theorem 7.2~\cite{doi:10.1137/1.9781611976465.160}]
   \label{theorem:ktoone1}
   Let $T = \{\tau_1, \ldots , \tau_n\} \subset X^{\leq m} $, $ \alpha \in [1, \infty)$, $\beta \in  [1, \infty)$, $\delta \in (0, 1)$, and let $T'\subseteq T$ be an arbitrary subset such that $|T'|\geq {|T|}\cdot{\beta}^{-1}$. Suppose that   
   there is an algorithm \texttt{Candidates} that given as input  $T,\alpha,\beta,\delta$ outputs $C\subset X^{\leq \ell}$ such that with probability 
at least $1 - \delta$, $C$ contains a point sequence $c$ such that $\cost_{p}^q(T', c) \leq \alpha \cdot \cost_{p}^q(T', c^{\ast}) $, where $c^{\ast}$ is a restricted $p$-\mean of $T'$. 
%some dissimilarity measure $\mathrm{d}$. 

Then, there is an algorithm \texttt{$k$-clustering} that given as input $(T, \emptyset, k, \beta, \delta)$, where $\beta \in (2k, \infty)$, $\delta \in (0, 1)$, $p,q\in [1,\infty)$, returns with
probability at least $1 - \delta$ a set $C = \{c_1, \ldots , c_k\}\subset X^{\leq \ell}$ with $\cost_{p}^q(T, C) \leq  \left(1 + \frac{4k}{\beta-2k}\right)\cdot \alpha \cdot \cost_{p}^q (T,C^{\ast})$,
where $C^{\ast}$ is an optimal solution to the $(k,\ell,p,q)$-clustering problem with input $T$.
%set of $k$-medians of $T$, under $\dtw_p$ raised to the power of $q$..
\end{theorem}
\begin{theorem}[Theorem 7.3~\cite{doi:10.1137/1.9781611976465.160}]\label{theorem:ktoone2}
   Let $T_1(n,\alpha, \beta, \delta$) denote the worst-case running time of \texttt{Candidates} for
an arbitrary input-set $T \subset X^{\leq m}$ with $|T| = n$ and let $C(n, \alpha, \beta, \delta)$ denote the maximum number of candidates
it returns. %Also, let $T_d$ denote the worst-case running time needed to compute the dissimilarity measure  for an input element and a candidate. 
If $T_1$ and $C$ are non-decreasing in $n$, then \texttt{$k$-clustering} has running time in
\[O\left(C(n, \alpha,\beta, \delta)^{k+2}
\cdot  n  m  \cdot T_{\rho} +
C(n, \alpha,\beta, \delta)^{k+1}\cdot  T_1(n, \alpha,\beta, \delta)\right),\] where $T_{\rho}$ denotes the worst-case running time needed to compute the distance between two points in $X$.
\end{theorem}

\subsubsection{$(k,\ell,p,p)$-Clustering}
In this section, we apply the result of \cref{section:constantfactor} to design a randomized algorithm for the $(k,\ell,p,p)$-clustering problem. 
The following algorithm is an adaptation of \cref{alg:1_median_2_3}. 
\begin{algorithm}[H]
\caption{$(1,\ell,p,p)$-clustering approximate candidates \label{alg:two_approximate_median_candidates}}
    \begin{algorithmic}[1]
        \Procedure{Cand1}{$T = \{\tau_1 = (\tau_{1,1}, \dots, \tau_{1,\lvert \tau_1 \rvert}), \dots, \tau_n = (\tau_{n,1}, \dots, \tau_{n,\lvert \tau_n \rvert}) \},\beta, \delta, \epsilon, p$}
            \State $P \gets \bigcup_{i=1}^n \bigcup_{j=1}^{\lvert \tau_i \rvert} \{ \tau_{i,j} \}$
            \State $S \gets$ sample $\left\lceil {(2^p{\epsilon^{-1}}+1) \beta m}  \ln\left({\ell}/{\delta}\right)\right\rceil$ points from $P$ uniformly and independently at 
            
            \hspace{\algorithmicindent} random  with replacement
            \State \Return $ S^{\leq \ell}$
        \EndProcedure
    \end{algorithmic}
\end{algorithm}

\begin{lemma}
\label{lemma:kmedians_constantapprox}
Let $T\subset X^{\leq m}$, $\beta >1$,  $\delta \in (0,1)$, $\epsilon>0$, $p\geq 1$. 
Let $T'\subseteq T$ such that $|T'|\geq |T|\cdot \beta^{-1}$ and  %let $S\subseteq T$ be a set obtained by sampling $\lceil 2\beta \cdot  \log(2/\delta)\rceil$ sequences from $T$ uniformly and independently at random with replacement. 
let $C$ be a set obtained by running \cref{alg:two_approximate_median_candidates} with input $(T,\beta,p,\delta)$. 
Let $c$ be an optimal $p$-\mean of $T'$. 
With probability at least $1-\delta$, there exists $\tau'\in C$ 
such that 
\[\cost_p^1(T',\tau')\leq (2^p+\epsilon) \cdot \cost_p^1 (T',c).\]
\end{lemma}
\begin{proof}
Let $c \in X^{\ell'}$, $\ell'\leq \ell$. 
By Lemma~\ref{lem:discrete_approx} applied on $T'$,  we have that there exist sets $B_1,\ldots,B_{\ell'}\subseteq P$, each of cardinality at least $\left(\frac{\epsilon}{2^p+\epsilon}\right)\cdot |T'|$ such that  any point sequence $c'=(c_1',\ldots,c_{\ell'}')$ with $\forall i\in [\ell]:~c_i'\in B_i$, is a $(2^p+\epsilon)$-approximate restricted $p$-\mean of $T'$. We upper bound the probability that $S$ does not contain any point from a fixed $B_i$:
\[
\Pr\left[\forall x\in B_i:~x\notin S \right] \leq \left(\frac{|P|- \left(\frac{\epsilon}{2^p+\epsilon}\right)\cdot |T'|}{|P|}\right)^{|S|}\leq 
\left(1- \left(\frac{\epsilon}{(2^p+\epsilon)\beta m}\right) \right)^{|S|} \leq \frac{\delta}{\ell}
\]
Then, by a union bound we have that the probability that there exists $i\in [\ell']$ such that $\forall x\in B_i:~ x\notin S$, is at most $\delta$. Hence, with probability at least $1-\delta$, there is a point sequence  $c'\in S^{\leq \ell}$ which is a $(2^p+\epsilon)$-approximate restricted $p$-\mean of $T'$.
\end{proof}

\begin{theorem}\label{theorem:kmedian1}
There is an algorithm that given a set $T\subset X^{\leq m}$ of $n$ point sequences, $p\in[1,\infty)$, $\beta \in (2k, \infty)$ and $\delta \in (0, 1)$, returns with
probability at least $1 - \delta$ a set $C = \{c_1, \ldots , c_k\}$ with $\cost_p^1(T, C) \leq  \left(1 + \frac{4k}{\beta-2k}\right)\cdot (2^p+\epsilon) \cdot \cost_p^1 (T,C^{\ast})$,
where $C^{\ast}$ is an optimal set of $k$-medians of $T$, under $\dtw_p$. The algorithm has running time in 
\[O\left(({(2^p{\epsilon^{-1}}+1) \beta m}  \ln\left({\ell}/{\delta}\right))^{\ell(k+2)}
\cdot  n  m   \right),\] assuming that the time needed to compute the distance between two points of $X$ is constant.
\end{theorem}
\begin{proof}
    We plug \cref{alg:two_approximate_median_candidates} into \cref{theorem:ktoone1,theorem:ktoone2}. By \cref{lemma:kmedians_constantapprox}, for any $T'\subseteq T$ with $|T'|\geq |T|\beta^{-1}$ \cref{alg:two_approximate_median_candidates} returns a set of point sequences which contains a $(2^p+\epsilon)$-approximate restricted $p$-\mean of $T'$. Therefore, by \cref{theorem:ktoone1} the clustering algorithm is correct. 
    The running time of \cref{alg:two_approximate_median_candidates} is upper bounded by $O(|S|^{\ell}+nm)$. The running time then follows by \cref{theorem:ktoone2}.
\end{proof}

\subsubsection{$k$-medians under $p$-DTW}
\label{section:kmedians}
In this section, we apply the random sampling bound developed in Section~\ref{section:simpltrineq} to design a randomized algorithm for the $(k,\ell,p,1)$-clustering problem, that is the problem of computing $k$-medians of complexity at most $\ell$, under $\dtw_p$. We achieve an approximation factor in $O(m^{1/p}\ell^{1/p})$. 
%\subsection{Computing candidates}

The main idea is that one can use random sampling and approximate simplifications, to obtain a simple algorithm for computing a set of $1$-median candidates. Those candidates are guaranteed, up to some user-defined probability, to contain a point sequence which is an approximate $1$-median for a fixed but unknown subset of the input. 
%Then, we are able to combine with previous results, to obtain an algorithm for the $k$-medians problem.  
\begin{algorithm}[H]
\caption{$(1,\ell,p,1)$-clustering approximate candidates \label{alg:alpha_apprx_median_candidates}}
    \begin{algorithmic}[1]
        \Procedure{Cand2}{$T = \{ \tau_1 , \dots, \tau_n  \}, \beta, p, \delta $}
            \State $S\gets$ sample $\lceil 2\beta \cdot  \log(2/\delta)\rceil$ point sequences from $T$ uniformly and independently at
            
\hspace{\algorithmicindent}            random  with replacement
           \State $C\gets \emptyset$
           \For{\textbf{each} $\tau \in S$}
           \State $\tau' \gets $ $(2,\ell)$-simplification of $\tau$, under $\dtw_p$
           \State $C \gets C \cup \{\tau' \}$
           \EndFor
           \State \Return $C$
           
    \EndProcedure
    \end{algorithmic}
\end{algorithm}

\begin{lemma}
\label{lemma:kmedians_roughapprox}
Let $T\subset X^{\leq m}$, $\beta >1$, $p\geq 1$, $\delta \in (0,1)$. 
Let $T'\subseteq T$ such that $|T'|\geq |T|\cdot \beta^{-1}$ and  %let $S\subseteq T$ be a set obtained by sampling $\lceil 2\beta \cdot  \log(2/\delta)\rceil$ sequences from $T$ uniformly and independently at random with replacement. 
let $C$ be a set obtained by running \cref{alg:alpha_apprx_median_candidates} with input $(T,\beta,p,\delta)$. 
Let $c$ be an optimal restricted $(p,1)$-\mean of $T'$. 
With probability at least $1-\delta$, there exists $\tau'\in C$ 
such that 
\[\cost_p^1(T',\tau')\leq 8 \cdot  m^{1/p}\ell^{1/p} \cdot \cost_p^1 (T',c).\]% where $\tau'$ is a $2$-approximate minimum-error-$\ell$-simplification of $\tau$.
\end{lemma}
\begin{proof}
We use a standard Chernoff bound (see \cite[Theorem 4.5]{probability_and_computing})  to upper bound the probability that $|S\cap T'|\leq  |S|/(2\beta)$. Notice that $\pexpected\left[ |S\cap T'|\right]\geq |S|\cdot \beta^{-1}$. Hence, 
\begin{align}\label{ineq:event1}
    \Pr\left[|S\cap T'| \leq \frac{|S|}{2\beta}\right]\leq
\exp\left( -\frac{|S|}{8\beta}\right) \leq \frac{\delta}{2}.
\end{align}

%\todo[inline]{Maybe define $S^\prime$ to be the subset of $T^\prime$ of size exactly $\lvert S \rvert/(2\beta)$ that is used in the algorithm?}
Let $\mathcal{E}_{T'}$ be the event that $|S\cap T'| > \frac{|S|}{2\beta}$. We condition the rest of the proof on the event $\mathcal{E}_{T'}$. Let $\tau'$ be a $(2,\ell)$-simplification of any point sequence $\tau \in S\cap T'$. Then, by \cref{lem:sampleonesequence},

\[
  \pexpected_{\tau} \left[\cost_p^1(T',\tau')\mid \mathcal{E}_{T'}\right]\leq 4m^{1/p}\ell^{1/p} \cdot \cost_p^1 (T',c).
\]
By Markov's inequality, 
$
\Pr\left[\cost_p^1(T',\tau') \geq 8m^{1/p}\ell^{1/p} \cdot \cost_p^1 (T',c) \mid \mathcal{E}_{T'}\right]\leq \frac{1}{2}$. 
Hence, by independence of the random sampling, 
\begin{align}
\label{ineq:event2}
    \Pr\left[ \forall \tau \in S\cap T': \cost_p^1(T',\tau') \geq 8m^{1/p}\ell^{1/p} \cdot \cost_p^1 (T',c) \mid \mathcal{E}_{T'} \right] \leq \frac{1}{2^{|S|/(2\beta)}} \leq \frac{\delta}{2}.
\end{align}
A union bound using inequalities (\ref{ineq:event1}), (\ref{ineq:event2}) completes the proof. 

\end{proof}

%\says{dennis}{d should be $\rho$, this is a little mistake in the paper and $\epsilon$ is not really needed}

\begin{theorem}\label{theorem:kmedian2}
There is an algorithm that given a set $T\subset X^{\leq m}$ of $n$ point sequences, $p\in[1,\infty)$, $\beta \in (2k, \infty)$ and $\delta \in (0, 1)$, returns with
probability at least $1 - \delta$ a set $C = \{c_1, \ldots , c_k\}$ with $\cost_p^1(T, C) \leq  \left(1 + \frac{4k}{\beta-2k}\right)\cdot (8m^{1/p}\ell^{1/p}) \cdot \cost_p^1 (T,C^{\ast})$,
where $C^{\ast}$ is an optimal set of $k$-medians of $T$, under $\dtw_p$. The algorithm has running time in %\todo{simplify this} %$O\left(\beta\log(1/\delta))^{k+2}\cdot  n  m\ell +(\beta \log(1/\delta))^{k+1}\cdot  (nm+\beta m^4 \ell \log(1/\delta) )\right)$. 
\[O\left((\beta \log(1/\delta))^{k+2}
\cdot  n  m  +
(\beta \log(1/\delta))^{k+1}\cdot\beta m^4  \log(1/\delta) \right),\]
assuming that the distance between two points of $X$ can be computed in constant time.
\end{theorem}
\begin{proof}
We plug \cref{alg:alpha_apprx_median_candidates} into \cref{theorem:ktoone1,theorem:ktoone2}. \cref{lemma:kmedians_roughapprox} guarantees that with probability at least $1-\delta$, there exists a point sequence in the set $C$, returned by  \cref{alg:alpha_apprx_median_candidates}, which is an $(8m^{1/p}\ell^{1/p})$-approximate $1$-median to an arbitrary subset $T'\subset T$, as required by \cref{theorem:ktoone1}. 
Let $T_{\rho}$ be time needed to compute the distance between two elements of $X$. 
Using \cref{theorem:simplification} to compute simplifications, \cref{alg:alpha_apprx_median_candidates} needs $O(T_{\rho}\cdot \beta m^4 \ell \log(1/\delta) ))$ time to compute $C$. Taking into account the time needed to read the input, and assuming that $d,\ell$ are constants, the total running time of \cref{alg:alpha_apprx_median_candidates} is in $O(nm+T_{\rho}\cdot\beta m^4  \log(1/\delta) )$. 
Therefore, by \cref{theorem:ktoone1}, there is an algorithm that returns an $\left(1 + \frac{4k}{\beta-2k}\right)\cdot (8m^{1/p}\ell^{1/p})$-approximate solution to the $k$-medians problem, and by \cref{theorem:ktoone2} the algorithm has running time in $O\left(\beta \log(1/\delta))^{k+2}
\cdot  n  m \cdot T_{\rho} +
(\beta \log(1/\delta))^{k+1}\cdot  (nm+T_{\rho}\cdot\beta m^4  \log(1/\delta) )\right)$. 
 \end{proof}

\end{document}